\documentclass[11pt,a4paper]{article}

\usepackage{amsmath}
\usepackage{graphicx}
\usepackage{amssymb}
\usepackage{bold-extra}
\usepackage{algorithm, algpseudocode}
\usepackage{color}
\usepackage{fullpage}
\usepackage{amsthm}

\newtheorem{thm}{Theorem}[section]{\bfseries}{\itshape}
\newtheorem{lma}[thm]{Lemma}{\bfseries}{\itshape}
{\bfseries}{\itshape}
{\bfseries}{\itshape}
{\bfseries}{\itshape}
\newtheorem{prop}[thm]{Proposition}{\bfseries}{\itshape}
\newtheorem{adef}{Definition}{\bfseries}{\itshape}

\DeclareMathOperator{\sig}{sig}
\DeclareMathOperator{\csig}{csig}
\DeclareMathOperator{\del}{del}
\DeclareMathOperator{\cdel}{cdel}
\DeclareMathOperator{\st}{st}
\DeclareMathOperator{\cst}{cst}

\DeclareMathOperator{\cintr}{c-intr}
\DeclareMathOperator{\fgt}{fgt}
\DeclareMathOperator{\cfgt}{c-fgt}
\DeclareMathOperator{\join}{join}
\DeclareMathOperator{\cjoin}{c-join}
\DeclareMathOperator{\image}{Im}
\DeclareMathOperator{\dom}{Dom}

\newcommand{\compprobname}{$\mathcal{T}_{h+1}$-\textsc{Free Edge Deletion}}
\newcommand{\genprobname}{$\mathcal{F}$-\textsc{Free Edge Deletion}}
\newcommand{\genruntime}{$2^{O(|\mathcal{F}|w^r)}n$}

\begin{document}

\title{Deleting edges to restrict the size of an epidemic: a new application for treewidth}
\author{Jessica Enright\thanks{Computing Science and Mathematics, University of Stirling \texttt{jae@cs.stir.ac.uk}} \and Kitty Meeks\thanks{School of Computing Science, University of Glasgow \texttt{kitty.meeks@glasgow.ac.uk}}}
\date{April 2017}
\maketitle

\maketitle

\begin{abstract}
Motivated by applications in network epidemiology, we consider the problem of determining whether it is possible to delete at most $k$ edges from a given input graph (of small treewidth) so that the resulting graph avoids a set $\mathcal{F}$ of forbidden subgraphs; of particular interest is the problem of determining whether it is possible to delete at most $k$ edges so that the resulting graph has no connected component of more than $h$ vertices, as this bounds the worst-case size of an epidemic.  While even this special case of the problem is NP-complete in general (even when $h=3$), we provide evidence that many of the real-world networks of interest are likely to have small treewidth, and we describe an algorithm which solves the general problem in time \genruntime ~on an input graph having $n$ vertices and whose treewidth is bounded by a fixed constant $w$, if each of the subgraphs we wish to avoid has at most $r$ vertices.  For the special case in which we wish only to ensure that no component has more than $h$ vertices, we improve on this to give an algorithm running in time $O((wh)^{2w}n)$, which we have implemented and tested on real datasets based on cattle movements.
\end{abstract}

\section{Introduction}

Network epidemiology seeks to understand the dynamics of disease spreading over a network or graph, and is an increasingly popular method of modelling real-world diseases.  The rise of network epidemiology corresponds to a rapid increase in the availability of contact network datasets that can be encoded as networks or graphs: typically, the vertices of the graph represent agents that can be infected and infectious, such as individual humans or animals, or appropriate groupings of these, such as cities, households, or farms. The edges are then the potentially infectious contacts between those agents.   Considering the contacts within a population as the edges of a graph can give a large improvement in disease modelling accuracy over mass action models, which assume that a population is homogeneously mixing.  For example, if we consider a sexual contact network in which the vertices are people and the edges are sexual contacts, the heterogeneity in contacts is very important for explaining the pattern and magnitude of an AIDS epidemic \cite{andersonAIDS}.

Our work has been especially motivated by the idea of controlling diseases of livestock by preventing disease spread over livestock trading networks. As required by European law, individual cattle movements between agricultural holdings in Great Britain are recorded by the 
British Cattle Movement Service (BCMS) \cite{bcmsDescription}; in early 2014, this dataset contained just under 300 million trades and  just over 133,000 agricultural holdings.  
For modelling disease spread across the British cattle industry, it is common to create vertices from farms, and edges from trades of cattle between those farms: a disease incursion starting at a single farm could spread across this graph through animal trades, as is thought to have happened during the economically-damaging 2001 British foot-and-mouth disease crisis \cite{fmdNetwork}.  

We are interested in controlling or limiting the spread of disease on this sort of network, and so have focussed our attention on edge deletion, which might correspond to forbidden trade patterns or, more reasonably, extra vaccination or disease surveillance along certain trade routes.  Introducing extra controls of this kind is costly, so it is important to ensure that this is done as effectively as possible.  Many properties that might be desirable from the point of view of restricting the spread of disease can be expressed in terms of forbidden subgraphs: edge-deletion to achieve a maximum degree of at most $d$ is equivalent to edge-deletion to a graph avoiding the star $K_{1,d+1}$, and edge-deletion to maximum component size at most $h$ is equivalent to edge-deletion to a graph avoiding all trees on $h+1$ vertices.  Clearly we can also combine criteria of this kind, for example edge-deletion to a graph which has maximum component size at most $h$ and maximum degree at most $d$.  We are therefore concerned with the following general problem.
\\

\hangindent=\parindent
$\mathcal{F}$-\textsc{Free Edge Deletion}\\
\textit{Input:} A graph $G = (V,E)$ and an integer $k$.\\
\textit{Question:} Does there exist $E' \subseteq E$ with $|E'| = k$ such that $G \setminus E'$ does not contain any $F \in \mathcal{F}$ as a subgraph?\\

This is in fact a special case of the more general problem in which we seek to avoid a set $\mathcal{F}$ of graphs as \emph{induced} subgraphs (which corresponds to edge-deletion to a hereditary class of graphs).  We have chosen to focus on the special case of edge-deletion to a monotone class of graphs (that is, a class closed under deletion of vertices and edges) as it is reasonable to assume in epidemiological applications that if we wish to avoid some subgraph $F$ then we also wish to avoid any graph $F'$ obtained from $F$ by adding edges.  Moreover, this assumption improves the running time of the algorithm by decreasing the size of the family $\mathcal{F}$ we wish to avoid, compared with expressing our target class in terms of forbidden induced subgraphs (for example, only one forbidden subgraph is required to define the class of graphs with maximum degree at most $d$, but to express this in terms of forbidden induced subgraphs we would have to forbid every induced subgraph on $d + 2$ vertices that contains a vertex of degree $d+1$).  However, it is straightforward to adapt the algorithm described in Section \ref{sec:gen-alg} to consider induced subgraphs.  The algorithm can also easily be adapted to deal with different costs associated with the deletion of different edges (so as to decide whether it is possible to delete edges with a total cost of at most $k$ to remove all copies of subgraphs from $\mathcal{F}$).

A special case of particular interest is the situation in which $\mathcal{F}$ is the set of all trees on $h+1$ vertices, so that we are deleting edges in order to obtain a graph in which every connected component contains at most $h$ vertices; $h$ is then an upper bound on the number of vertices which may, in the worst case, be infected from a single initially infected vertex.  We denote by $\mathcal{T}_{h+1}$ the set of all trees on $h+1$ vertices, so this special case is the problem \compprobname.  For this case, we also consider two straightforward extensions of the problem that are relevant for real-world applications: 
\begin{itemize}
\item assigning different weights to different vertices (e.g. corresponding to the number of animals in a particular animal holding), and seeking to bound the total weight of each connected component;
\item imposing different limits on the size of components containing individual vertices (for example, we might want to enforce a smaller size limit for components containing certain vertices considered to be of particularly high risk).
\end{itemize}

Even the special case \compprobname ~of our general problem is intractable in general for constant $h$, as demonstrated in the following proposition.  The reduction relies on the observation that the maximum number of edges in a graph having maximum component size $h$ is obtained if the graph is a disjoint union of $h$-cliques. 
\begin{prop}
	\compprobname ~is NP-complete for every $h \geq 3$.
	\label{prop:hardness}
\end{prop}
\begin{proof}
	First, observe that an edge set of size $k$ to be deleted suffices as a certificate for this problem, therefore \compprobname ~is in NP. We prove NP-hardness by means of a reduction from the following problem, shown to be NP-hard in \cite{gareyJohnson}:\\
	
	\hangindent=\parindent
	\textsc{Perfect Triangle Cover}\\
	\textit{Input:} A graph $G = (V,E)$.\\
	\textit{Question:} Does there exist a set of vertex-disjoint triangles that cover all vertices in the graph?\\
	
	\hangindent=0pt
	
	Starting with an instance $G = (V, E)$ of \textsc{Perfect Triangle Cover} (where $G$ has $n$ vertices), we will produce an instance of $\mathcal{T}_{h+1}$-\textsc{Free Edge Deletion} for arbitrary $h \geq 3$ that is a \textbf{yes} instance if and only if $G$ is a \textbf{yes} instance of \textsc{Perfect Triangle Cover}.
	
	We do this via an intermediate problem, which is a generalisation of \textsc{Perfect Triangle Cover} to perfect arbitrarily-sized clique covers: \\
	
	\hangindent=\parindent
	\textsc{Perfect $K_h$ Cover}\\
	\textit{Input:} A graph $G = (V,E)$.\\
	\textit{Question:} Does there exist a set of vertex-disjoint cliques of size $h$ that cover all vertices in the graph?\\
	
	\hangindent=0pt
	
	Note that this problem is in NP: a set of covering cliques constitutes a certificate.  In all instances of \textsc{Perfect $K_h$ Cover}, we will assume the graph has a number of vertices divisible by $h$: otherwise this is trivially a \textbf{no} instance.  

	We show how to transform an instance of \textsc{Perfect $K_h$ Cover} to \textsc{Perfect $K_{h+1}$ Cover}: let $G = (V, E)$ be an instance of \textsc{Perfect $K_h$ Cover}: produce graph $G'$  by adding to $G$ an independent set of size $\frac{|V|}{h}$, with every element of the independent set adjacent to all vertices in $V$.  We claim that $G'$ is a \textbf{yes} instance of \textsc{Perfect $K_{h+1}$ Cover} if and only if $G$ is a \textbf{yes} instance of \textsc{Perfect $K_h$ Cover}.  First, suppose that $G'$ is a \textbf{yes} instance of \textsc{Perfect $K_{h+1}$ Cover}: then the intersection with $V$ of the $K_{h+1}$ that perfectly cover the vertices of $G'$ are a set of $K_h$ that perfectly cover the vertices of $G$.  Conversely, suppose that $G$ is a \textbf{yes} instance of \textsc{Perfect $K_h$ Cover}, then if we extend each of the $K_h$ that perfectly cover the vertices of $G$ by exactly one of the vertices in the new independent set, we have a set of $K_{h+1}$ that perfectly cover the vertices of $G'$.
	
	Note that \textsc{Perfect Triangle Cover} is exactly \textsc{Perfect $K_3$ Cover}: then by iteration we have that \textsc{Perfect $K_h$ Cover} is NP-complete for any $i \geq 3$. 
	
	We now reduce \textsc{Perfect $K_h$ Cover} to $\mathcal{T}_{h+1}$-\textsc{Free Edge Deletion}.  If $G = (V, E)$ is an instance of \textsc{Perfect $K_h$ Cover},  let $G = (V, E)$ and $k = |E| - \frac{1}{2}(h-1)n$  be an instance of $\mathcal{T}_{h+1}$-\textsc{Free Edge Deletion}; we claim that $G$ is a \textbf{yes} instance of \textsc{Perfect $K_h$ Cover} if and only if $(G,k)$ is a \textbf{yes} instance of $\mathcal{T}_{h+1}$-\textsc{Free Edge Deletion}.
	
	Suppose first that $G = (V, E)$ is a \textbf{yes} instance of \textsc{Perfect $K_h$ Cover}, so there exists a set of vertex-disjoint $K_h$ that cover all vertices of $G$; let $E'$ be the set of edges induced by those $K_h$.  Note that $E'$ contains exactly $\frac{n}{h}\binom{h}{2} =  \frac{1}{2}(h-1)n$ edges, so the edge-set $E \backslash E'$ is of size $|E| - \frac{1}{2}(h-1)n = k$; moreover, as there are no edges in $E'$ between distinct $K_h$, the graph $G \backslash (E \backslash E')$ contains no connected component on more than $h$ vertices.  Thus $(G,k)$ is a \textbf{yes} instance of $\mathcal{T}_{h+1}$-\textsc{Free Edge Deletion}.
	
	Conversely, suppose that $(G,k)$ is a \textbf{yes} instance of $\mathcal{T}_{h+1}$-\textsc{Free Edge Deletion}. Then there exists an edge set $F$ such that $|F| = k = |E| - \frac{1}{2}(h-1)n$ and  every connected component of $G \backslash F$ is of size $h$ or less; by a pigeonhole argument, these components must all be $K_h$.  Therefore $G$ is a \textbf{yes} instance of  \textsc{Perfect $K_h$ Cover}, as claimed. 

\end{proof}

In order to develop useful algorithms for real-world applications, we therefore need to exploit structural properties of the input network.  In Section \ref{sec:real-tw} we provide evidence that many animal trade networks of interest are likely to have small treewidth, and in Section \ref{sec:gen-alg} we describe an algorithm to solve \genprobname ~whose running time on an $n$-vertex graph of treewidth $w$ is bounded by \genruntime, if every graph in $\mathcal{F}$ has at most $r$ vertices; this algorithm is easily adapted to output an optimal solution.  In Section \ref{sec:cpt-alg} we then improve on this for the special case of \compprobname, describing an algorithm to solve this problem in time $O((wh)^{2w}n)$ on an $n$-vertex graph of treewidth $w$.  Many problems that are thought to be intractable in general are known to admit polynomial-time algorithms when restricted to graphs of bounded treewidth, often by means of a dynamic programming strategy similar to that used to attack the problem considered here; however, to the best of the authors' knowledge, the usefulness of such algorithms for solving real-world network problems has yet to be investigated thoroughly. 

In reality, policy decisions about where to introduce controls are likely to be influenced by a range of factors, which cannot all be captured adequately in a network model.  Thus, the main application of our algorithm will be in comparing any proposed strategy with the theoretical optimum: a policy-maker can determine whether there is a solution with the same total cost that results in a smaller maximum component size.  We provide an example of an experimental application of our algorithm to cattle trading networks in Section \ref{sec:expt}.

In the remainder of this section, we begin by reviewing previous related work in Section \ref{sec:previous} before introducing some important notation in Section \ref{sec:notation} and reviewing the key features of tree decompositions in Section \ref{sec:treewidth}.

\subsection{Review of previous work}
\label{sec:previous}

From a combinatorial perspective, we are concerned here with \emph{edge-deletion problems}.  An edge-deletion problem asks if there is a set of at most $k$ edges that can be deleted from an input graph to produce a graph in some target class.  In contrast to the related well-characterised vertex-deletion problems \cite{Lewis80}, there is not yet a complete characterisation of the hardness of edge-deletion problems by target graph class.

Yannakakis \cite{yannaSTOC1978} gave early results in edge-deletion problems, showing that edge-deletion to planar graphs,
outer-planar graphs, line graphs, and  transitive digraphs is NP-complete.  Subsequently, Watanabe, Ae and Nakamura \cite{watanabe1983} showed that edge-deletion problems are NP-complete if the target graph class can be finitely characterised by 3-connected graphs.  There are a number of further hardness results known for edge-deletion to well-studied graph classes, including for interval and unit interval graphs \cite{goldberg1993}, cographs \cite{elmallah1988}, and threshold graphs \cite{margot1994} and, as noted in \cite{natanzon2001}, hardness of edge-deletion to bipartite graphs follows from the hardness of a MAX-CUT problem.  Natanzon, Shamir and Sharan \cite{natanzon2001} further showed NP-completeness of edge-deletion to disjoint unions of cliques, and perfect, chain, chordal, split, and asteroidal-triple-free graphs, but also give polynomial-time algorithms, in the special case of the input graph having bounded degree, for edge-deletion to chain, split, and threshold graphs.  

Given the large number of hardness results in the literature, it is natural to consider the parameterised complexity of these problems.  Cai \cite{caiApprox} initiated this investigation, showing that edge-deletion to a graph class characterisable by a finite set of forbidden \emph{induced} subgraphs is fixed-parameter tractable when parameterised by $k$ (the number of edges to delete): he gave an algorithm to solve the problem in time $O(r^{2k}\cdot n^{r+1})$, where $n$ is the number of vertices in the input graph and $r$ is the maximum number of vertices in a forbidden induced subgraph.  Further fpt-algorithms have been obtained for edge-deletion to split graphs \cite{ghoshSplit2012} and to chain, split, threshold, and co-trivially perfect graphs \cite{guoKernels}.  When considering graphs of small treewidth, our algorithm for \genprobname ~(and indeed its adaptation to deal with forbidden induced subgraphs) represents a significant improvement on Cai's algorithm, with our running time of \genruntime.  While the fixed parameter tractability of this problem (parameterised by $r$, the maximum number of vertices in any element of $\mathcal{F}$) restricted to graphs of bounded treewidth does follow from the optimization version of Courcelle's Theorem \cite{arnborg91,courcelle93}, this does not lead to a practical algorithm for addressing real-world problems.  Note that Proposition \ref{prop:hardness} implies that parameterisation by $r$ alone will not be sufficient to give an fpt-algorithm. 

The specific problem of modifying a graph to bound the maximum component size has previously been studied both in the setting of epidemiology \cite{li11} and in the study of network vulnerability \cite{gross13,drange14}.  The edge-modification version we consider here appears in the literature under various names, including the \emph{component order edge connectivity problem} \cite{gross13} and the \emph{minimum worst contamination problem} \cite{li11}.  Li and Tang \cite{li11} show that it is NP-hard to approximate the minimisation version of the problem to within $2 - \epsilon$, while Gross et.~al.\cite{gross13} describe a polynomial-time algorithm to solve the problem when the input graph is a tree.

\subsection{Notation and definitions}
\label{sec:notation}

Unless otherwise stated, all graphs are simple, undirected, and loopless.  For graph $G = (V, E)$, $V=V(G)$ is the vertex set of $G$, and $E=E(G)$ the edge set of $G$.  We denote the sizes of the edge and vertex sets of $G$ as $e(G) = |E(G)|$ and $v(G) = |V(G)|$.  For a vertex $v \in V(G)$, we say that vertex $u \in V(G)$ is a \emph{neighbour} of $v$ if $(u, v) \in E(G)$, and write $N_G(v)$ for the set of neighbours of $v$ in $G$.  If $U \subseteq V(G)$, we write $G[U]$ for the subgraph of $G$ induced by the vertex-set $U$.  Given a graph $G = (V,E)$ and a vertex $v \in V(G)$, we write $G \setminus v$ for the graph $G[V \setminus \{v\}]$.  Given a set of edges $E' \subseteq E(G)$ and a vertex $v \in V(G)$, we write $E' \setminus v$ for the set of edges in $E'$ that are not incident with $v$.  For any vertex $v \in V(G)$, we write $d_G(v)$ for the degree of $v$ in $G$; when the graph $G$ is clear from the context we may omit the subscript.  For further general graph notation, we direct the reader to \cite{golumbicBook}.

Given two graphs $H_1$ and $H_2$ with $v(H_1) \leq v(H_2)$, an \emph{embedding} of $H_1$ into $H_2$ is an injective function $\theta: V(H_1) \rightarrow V(H_2)$ such that $\theta(u)\theta(v) \in E(H_2)$ whenever $uv \in E(H_1)$.  Thus the graph $G$ contains the graph $F$ as a subgraph if and only if there is an embedding of $F$ into $G$.  We say that the embedding $\theta$ is a \emph{strong embedding} (or \emph{induced embedding}) if we have $\theta(u)\theta(v) \in E(H_2)$ if and only if $uv \in E(H_1)$ (so $\theta$ preserves non-adjacency, as well as adjacency).  

Given a function $\theta$ which maps a subset of $X$ to a subset of $Y$, we write $\dom(\theta)$ and $\image(\theta)$ for the domain and image of $\theta$ respectively.  Given a subset $X' \subseteq \dom(\theta)$, we write $\theta|_{X'}$ for the restriction of $\theta$ to $X'$.

A \emph{partition} $\mathcal{P}$ of a set $X$ is a collection of disjoint, non-empty sets whose union is $X$.  We call each set in the partition a \emph{block} of the partition, and every partition corresponds to a unique equivalence relation on $X$ where $x \sim y$ if and only if $x$ and $y$ belong to the same block of $X$.  If $\mathcal{P}$ and $\mathcal{P}'$ are partitions of $X$, we say that $\mathcal{P}'$ \emph{refines} $\mathcal{P}$ if every block of $\mathcal{P}'$ is contained in a single block of $\mathcal{P}$.  If $\mathcal{P}$ is a partition of $X$, and $y \in X$, we write $\mathcal{P} \setminus y$ for the partition of $X \setminus \{y\}$ obtained by removing the occurrence of $y$ from $\mathcal{P}$ (and, if this results in an empty set in the partition, also removing this empty set).

\subsection{Tree decompositions}
\label{sec:treewidth}

In this section we review the concept of a tree decomposition (introduced by Robertson and Seymour in \cite{robertsonSeymourGM2}) and introduce some of the key notation we will use throughout the rest of the paper.

Given any tree $T$, we will assume that it contains some distinguished vertex $r(T)$, which we will call the \emph{root} of $T$.  For any vertex $v \in V(T) \setminus r(T)$, the \emph{parent} of $v$ is the neighbour of $v$ on the unique path from $v$ to $r(T)$; the set of \emph{children} of $v$ is the set of all vertices $u \in V(T)$ such that $v$ is the parent of $u$.  The \emph{leaves} of $T$ are the vertices of $T$ whose set of children is empty.  We say that a vertex $u$ is a \emph{descendant} of the vertex $v$ if $v$ lies somewhere on the unique path from $u$ to $r(T)$ (note therefore that every vertex is a descendant of the root).  Additionally, for any vertex $v$, we will denote by $T_v$ the subtree induced by $v$ together with the descendants of $v$.

We say that $(T,\mathcal{D})$ is a \emph{tree decomposition} of $G$ if $T$ is a tree and $\mathcal{D} = \{\mathcal{D}(t): t \in V(T)\}$ is a collection of non-empty subsets of $V(G)$ (or \emph{bags}), indexed by the nodes of $T$, satisfying:
\begin{enumerate}
\item $V(G) = \bigcup_{t \in V(T)} \mathcal{D}(t)$,
\item for every $e=uv \in E(G)$, there exists $t \in V(T)$ such that $u,v \in \mathcal{D}(t)$,
\item for every $v \in V(G)$, if $T(v)$ is defined to be the subgraph of $T$ induced by nodes $t$ with $v \in \mathcal{D}(t)$, then $T(v)$ is connected.
\end{enumerate}
The \emph{width} of the tree decomposition $(T,\mathcal{D})$ is defined to be $\max_{t \in V(T)} |\mathcal{D}(t)| - 1$, and the \emph{treewidth} of $G$ is the minimum width over all tree decompositions of $G$.

We will denote by $V_t$ the set of vertices in $G$ that occur in bags indexed by the descendants of $t$ in $T$.  Thus, $V_t = \bigcup_{t' \in V(T_t)} \mathcal{D}(t')$.  

Later in this paper, we will exploit two useful properties that follow from the definition of a tree decomposition:
\begin{enumerate}
\item If $v \in \mathcal{D}(t)$ and $t'$ is a child of $t$ with $v \notin \mathcal{D}(t')$, then any path in $G$ from $v$ to a vertex $w \in V_{t'}$ must include at least one vertex of $\mathcal{D}(t) \setminus \{v\}$.
\item If $t \in T$ and $t_1,t_2$ are children of $t$, then any path in $G$ from a vertex $v_1 \in V_{t_1}$ to a vertex $v_2 \in V_{t_2}$ must contain at least one vertex from $\mathcal{D}(t)$.
\end{enumerate}

Although it is NP-hard to determine the treewidth of an arbitrary graph \cite{arnborg87}, it is shown in \cite{bodlaender93} that the problem of determining whether a graph has treewidth at most $w$, and if so computing a tree-decomposition of width at most $w$, can be solved in linear time for any constant $w$ (although the running time depends exponentially on $w$).

\begin{thm}[Bodlaender \cite{bodlaender93}]
For each $w \in N$, there exists a linear-time algorithm, that tests whether a given graph $G = (V, E)$ has treewidth at most $w$, and if so, outputs a tree decomposition of $G$ with treewidth at most $w$.
\end{thm}

A special kind of tree decomposition, known as a \emph{nice tree decomposition}, was introduced by Kloks \cite{kloks94}.  The nodes in such a decomposition can be partitioned into four types (examples in Figure \ref{fig:nodeTypes}):
\begin{description}
\item[Leaf nodes:] $t$ is a leaf in $T$.
\item[Introduce nodes:] $t$ has one child $t'$, such that $\mathcal{D}(t') \subset \mathcal{D}(t)$ and $|\mathcal{D}(t)| = |\mathcal{D}(t')| + 1$.
\item[Forget nodes:] $t$ has one child $t'$, such that $\mathcal{D}(t') \supset \mathcal{D}(t)$ and $|\mathcal{D}(t)| = |\mathcal{D}(t')| - 1$.
\item[Join nodes:] $t$ has two children, $t_1$ and $t_2$, with $\mathcal{D}(t_1) = \mathcal{D}(t_2) = \mathcal{D}(t)$.
\end{description}

\begin{figure}[h!]
\centering
\scalebox{0.7}{\includegraphics{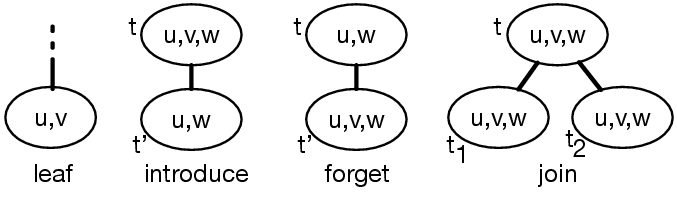}}
\caption{The four types of node in a nice tree decomposition.  From left to right: a leaf, an introduce node, a forget node, and a join node.}
\label{fig:nodeTypes}
\end{figure}

Any tree decomposition can be transformed into a nice tree decomposition in linear time:

\begin{lma}[\cite{kloks94}]
For constant $k$, given a tree decomposition of a graph $G$ of width $w$ and $O(n)$ nodes, where $n$ is the number of vertices of $G$, one can find a nice tree decomposition of $G$ of width $w$ and with at most $4n$ nodes in $O(n)$ time. 
\end{lma}

\section{Treewidth of real networks}
\label{sec:real-tw}

While the overall graph of cattle trades in Great Britain from 2001 to 2014 is fairly dense, many of the edges are repeated or parallel trades: that is, a farm sending animals over time to the same place, or many individual animals being moved at the same time; if we restrict our attention to a limited time frame, and ignore movements that would generate multiple edges (that is, we require our graph to be simple), the graph is quite sparse: for example, the aggregated graph of cattle trades in Scotland in 2009 has an edge-to-vertex ratio of approximately 1.15 when aggregating over January alone, and approximately 1.34 when aggregating over the entire year.  When considering an epidemic, it is much more relevant only to consider trades occurring within some restricted time frame (whose precise duration depends on the disease under consideration).  

We have investigated the treewidth of real networks arising from two kinds of cattle trade data.  First of all, for years between 2009 and 2014, we generated a graph from a type of persistent trade link recorded by BCMS in Scotland.  The largest of these is derived from the trades in 2013, and includes approximately 7,000 nodes and 6,000 edges (this lower density is typical when considering only persistent trade links, or trades over a restricted time period).  None of these graphs has treewidth more than fifteen, with most having treewidth less than four.  Secondly, in addition to these persistent trade links, we have computed an upper bound of the treewidth of the largest component of an aggregated, undirected version of the overall network of cattle trades in Scotland in 2009 over a variety of time windows, as illustrated in Figure \ref{fig:allScotlandTreewidth09}.  The treewidths of these components remains low even for large time windows: for an aggregation of all movements in a 200-day window the treewidth is below 10, and for all movements over the year it is below 18.  It is unlikely to be necessary to include a full year of movements in the analysis of any single epidemic, as the time scale of most exotic epidemics is much shorter.   

\begin{figure}
\centering
\includegraphics[width=0.7 \linewidth]{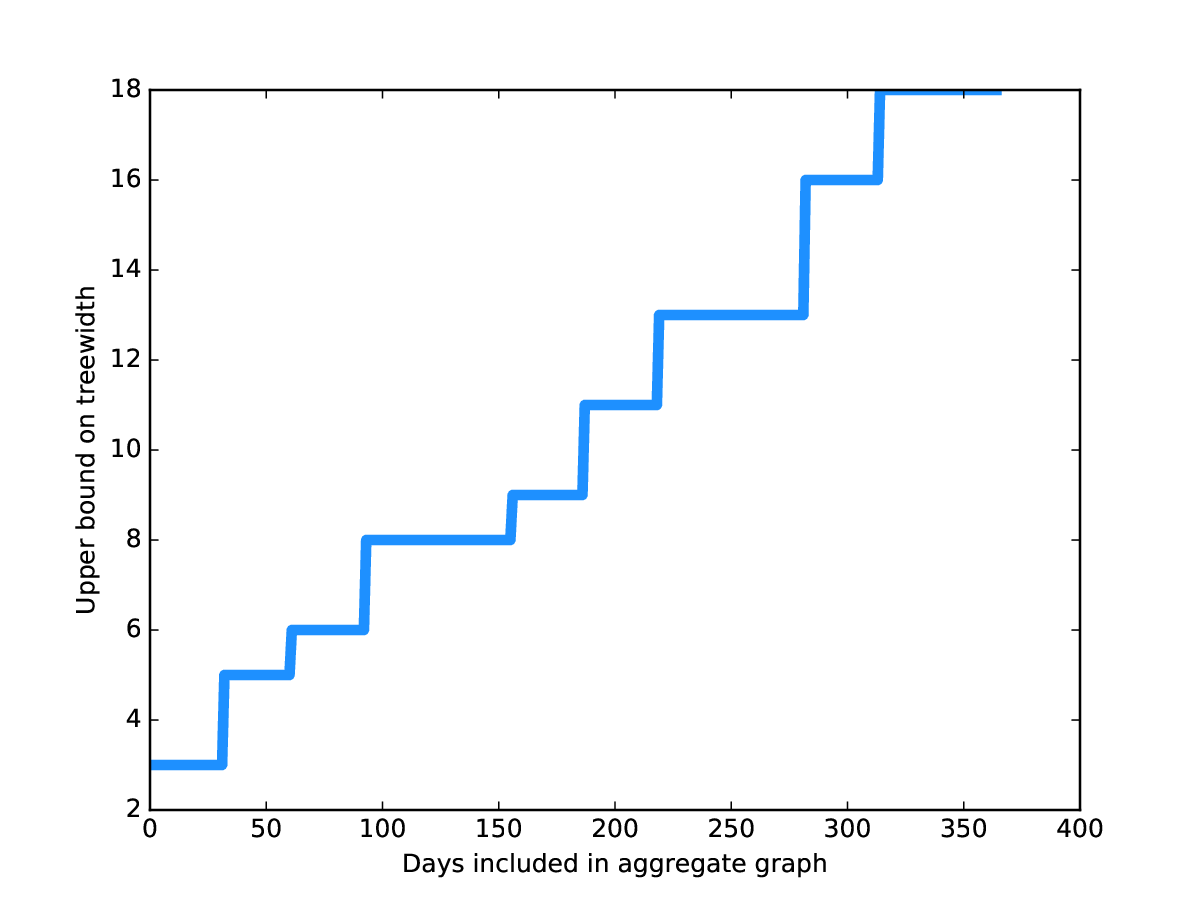}
\caption{A plot of an upper bound treewidth of the largest component in an aggregated, undirected version of the cattle movement graph in Scotland in 2009 over a number of different days included: all day sets start on January 1, 2009. Treewidths below eight are exact, treewidths over eight are upper bounds of the true treewidth.}
\label{fig:allScotlandTreewidth09}
\end{figure}

When modelling disease processes on graphs, it is common to model contact networks with any of a variety of random graph models, including Erd\H{o}s-R\'enyi random graphs or random graphs with scale-free degree distributions.  The cattle trading graphs considered here exhibit much lower treewidth than we would expect from an Erd\H{o}s-R\'enyi graph with the same density: Gao \cite{Gao2012566} has shown that Erd\H{o}s-R\'enyi graphs with an edge-to-vertex ratio of at least 1.073 (which our graphs satisfy), have treewidth linear in the number of vertices with high probability.  Scale-free random graphs have also been shown to have high treewidth: Gao also shows that scale-free graphs produced by preferential attachment with at least 11 edges added at each step have treewidth linear in the number of vertices with high probability; however, our graphs have lower density than this.

We have excluded markets from the data used above (so that if animals are sold from farm $A$ to farm $B$ via a market, this is considered as a direct trade from $A$ to $B$), so the low treewidth cannot be explained simply by the fact that most trades take place through a relatively small number of markets.  One possible explanation for the low treewidth is the structure of the industry: farms can sometimes be characterised by ``type'', with breeders producing calves who then might be grown at one or two other farms before eventual slaughter, meaning that cycles are unlikely to occur frequently in the network.  A more thorough investigation of underlying processes that lead to trade networks of low treewidth, and their influence on other graph parameters, seems a fruitful direction for future research.

While we have by no means completed an exhaustive study of the structural properties of real-world livestock trade networks, the evidence given here seems sufficient to suggest that algorithms which achieve a good running time on graphs of bounded treewidth will be useful for this application in practice.  Indeed, the benefits of exploiting the treewidth of the input graph are demonstrated by our initial experimental results in Section \ref{sec:expt}.

\section{The general algorithm}
\label{sec:gen-alg}

In this section, we describe an algorithm which, given a graph $G$ together with a nice tree decomposition $(T,\mathcal{D})$ of $G$ of width at most $w$, solves \genprobname ~on input $G$ in time \genruntime.  Since there exist linear-time algorithms both to compute a tree-decomposition of any graph $G$ of fixed treewidth $w$, and to transform an arbitrary tree decomposition into a nice tree decomposition, this in fact gives an algorithm which takes as input just a graph $G$ of treewidth at most $w$.  Thus, we prove the following theorem.

\begin{thm}
There exists an algorithm to solve \genprobname ~in time \genruntime ~on an input graph with $n$ vertices whose treewidth is at most $w$, if every element of $\mathcal{F}$ has at most $r$ vertices.
\label{thm:gen-alg}
\end{thm}

As with many algorithms that use tree decompositions, our algorithm works by recursively carrying out computations for each node of the tree, using the results of the same computation carried out on any children of the node in question.  In this case, we recursively compute the \emph{signature} of each node: we define the signature of a node in Section \ref{sec:gen-signature}.  It is then possible to determine whether we have a yes- or no-instance to the problem by examining the signature of the root of $T$.  

We may assume, without loss of generality, that every element of $\mathcal{F}$ contains at least one edge: if some $F \in \mathcal{F}$ has no edges, then any graph on at least $v(F)$ vertices will contain a copy of $F$ (no matter how many edges we delete) so there cannot be a yes-instance with $r$ or more vertices, and a brute-force approach will achieve the desired time bound on input graphs having fewer than $r$ vertices.

In Section \ref{sec:gen-recursive}, we describe how we compute the signature of a bag indexed by a given node from the signatures of its children, before discussing the running time and a number of extensions in Section \ref{sec:gen-runtime}.

\subsection{The signature of a node}
\label{sec:gen-signature}

In this section we describe the information we must compute for each node, and define the signature of a node.  Throughout the algorithm, we need to record the possible states corresponding to a given bag.  A valid \emph{state} of a bag $\mathcal{D}(t)$ is a pair consisting of:
\begin{enumerate}
\item a spanning subgraph $H$ of $G[\mathcal{D}(t)]$ which does not contain any $F \in \mathcal{F}$ as a subgraph, and
\item a function $\phi: X_{\mathcal{F},H} \rightarrow \{0,1\}$ (where $X_{\mathcal{F},H}$ consists of all pairs $(F',\theta)$ such that $F'$ is an induced subgraph of some $F \in \mathcal{F}$ and $\theta$ is an embedding into $H$ of some induced subgraph $F''$ of $F'$) which satisfies the following conditions:
\begin{enumerate}
\item $\phi(F',\theta) = 1$ whenever $\theta$ is an embedding of $F'$ into $H$;
\item if $\phi(F_1,\theta) = 1$ and $F_2$ is an induced subgraph of $F_1$ then $\phi(F_2,\theta|_{V(F_2)}) = 1$;
\item if $\phi(F_1,\theta) = 1$ (where $F_1$ is an induced subgraph of $F \in \mathcal{F}$ and $\theta$ is an embedding of $F_1'$ into $H$) then, if there exist $v \in V(H) \setminus \image(\theta)$ and $u \in V(F) \setminus V(F_1)$ such that, for every $w \in V(F_1)$ with $uw \in E(F)$,
\begin{itemize}
\item $w \in F_1'$, and
\item $\theta(w) v \in E(H)$, 
\end{itemize}
then $\phi(F[V(F_1) \cup \{u\}],\widehat{\theta}) = 1$, where $\widehat{\theta}$ extends $\theta$ by mapping $u$ to $v$.
\item for every $F \in \mathcal{F}$ and every $\theta$, $\phi(F,\theta) = 0$.
\end{enumerate}
\end{enumerate}

Intuitively, the state of a node tells us which forbidden subgraphs appear in $G[V_t]$ (and where these appear), so we know which partial forbidden subgraphs we must avoid extending to full copies of these graphs.  Conditions 2(b) and 2(c) are to ensure consistency, so that if a particular partial embedding of a forbidden subgraph is present then the appropriate extensions and restrictions of this embedding are also present.

For any bag $\mathcal{D}(t)$, we denote by $\st(t)$ the set of possible states of $\mathcal{D}(t)$.  Note that there are at most $2^{w^2}$ possible spanning subgraphs $H$, since $|\mathcal{D}(t)| \leq w+1$ (as the treewidth of $G$ is $w$) and $2^{\binom{w+1}{2}} \leq 2^{w^2}$.  For any given spanning subgraph $H$ and any $F''$ which is an induced subgraph of some $F \in \mathcal{F}$ (with $v(F'') \leq v(H)$), the number of embeddings of $F''$ into $H$ is at most the number of injective functions from $V(F'')$ to $V(H)$, which is equal to $\frac{v(H)!}{\left(v(H)-v(F'')\right)!} < v(H)^{v(F'')} \leq (w+1)^{v(F'')}$; summing over all possibilities for $F''$, we have that the total number of pairs $(F',\theta)$ such that $F''$ is an induced subgraph of $F'$ and $\theta$ is a an embedding of $F''$ into $H$ is at most
\begin{align*}
\sum_{F \in \mathcal{F}} \quad \sum_{\substack{F' \text{ an induced}\\ \text{subgraph of } F}} \quad \sum_{\substack{F'' \text{ an induced}\\ \text{subgraph of } F'}} (w+1)^{v(F'')}  & = \sum_{F \in \mathcal{F}} \sum_{i=1}^{v(F)} \sum_{j=1}^i \binom{v(F)}{i} \binom{i}{j} (w+1)^j \\
		& \leq \sum_{F \in \mathcal{F}} \sum_{i=1}^{v(F)} \binom{v(F)}{i} (w+2)^i \\
		& \leq \sum_{F \in \mathcal{F}} (w+3)^{v(F)} \\
		& \leq |\mathcal{F}| (w+3)^r.
\end{align*}
Not all such pairs will belong to $X_{\mathcal{F},H}$, as they may not obey the condition on edges, but this certainly provides an upper bound on $|X_{\mathcal{F},H}|$.  The number of possibilities for the function $\phi$ is therefore at most $2^{|\mathcal{F}| (w+3)^r}$.  Overall, this gives an upper bound of $2^{w^2 + |\mathcal{F}| (w+3)^r}$ on the number of valid states for each bag.  Note that, since every element of $\mathcal{F}$ has at most $r$ vertices, the cardinality of $\mathcal{F}$ is at most $\sum_{i=2}^r 2^{\binom{i}{2}} < 2^{r^2}$, so we can bound the number of valid states by $2^{w^2 + 2^{r^2} (w+3)^r}$.

Given any state, we make the following definition.
\begin{adef}
\label{def:gen-E}
Let $\sigma = (H, \phi) \in \st(t)$.  Then $\mathcal{E}(t,\sigma)$ is be the set of edge-sets $E' \subset E(G[V_t])$ such that $\widetilde{G_t} = G[V_t] \setminus E'$ has the following properties:
\begin{enumerate}
\item $\widetilde{G_t}[\mathcal{D}(t)] = H$,
\item $\widetilde{G_t}$ does not contain any $F \in \mathcal{F}$ as a subgraph, and
\item for every $F \in \mathcal{F}$ and every subgraph $F'$ of $F$, if $\theta$ is an embedding of $F'$ into $\widetilde{G_t}$ and $\theta(V(F')) \cap \mathcal{D}(t) \neq \emptyset$ then $\phi(F', \theta|_{\mathcal{D}(t)}) = 1$.
\end{enumerate}
\end{adef}
Note that, whenever $\sigma$ is a valid state for $t$, the set $\mathcal{E}(t,\sigma)$ will be non-empty: setting $E' = E(G[V_t])$ will always satisfy all three conditions, since every $F \in \mathcal{F}$ contains at least one edge.

Since we are interested in determining whether it is possible to delete at most $k$ edges to obtain a graph that does not contain any element of $\mathcal{F}$ as a subgraph, we will primarily be interested in a subset of $\mathcal{E}(t,\sigma)$: for any node $t$ and $\sigma \in \st(t)$ we define this subset as
$$\mathcal{E}_k(t,\sigma) = \{E' \in \mathcal{E}(t,\sigma): |E'| \leq k\}.$$
We then define
$$\del_k(t,\sigma) = \min_{E' \in \mathcal{E}_k(t,\sigma)} |E'|,$$
adopting again the convention that the minimum, taken over an empty set, is equal to infinity.  To simplify notation, given any $a,b \in \mathbb{N}$, we define $[a]_{\leq b}$ to be equal to $a$ if $a \leq b$, and equal to $\infty$ otherwise.  Finally, we define the \emph{signature} of a node $t$ to be the function $\sig_t: \st(t) \rightarrow \{0,1,\ldots,k,\infty\}$ such that $\sig_t(\sigma) = \del_k(t,\sigma)$.  

Our input graph is then a yes-instance to \genprobname ~if and only if there exists some $\sigma \in \st(t_r)$ such that $\sig_{t_r}(\sigma) \leq k$, where $t_r$ is the root of the tree indexing the decomposition.

\subsection{Computing signatures recursively}
\label{sec:gen-recursive}

In this section we describe how the signature of any given node in a nice tree decomposition can be determined from the signatures of its children (if any).  This relationship depends on the type of node under consideration, and we will again discuss each type of node in turn, describing how to compute $\del_k(t,\sigma)$ for an arbitrary state $\sigma$.

\subsubsection*{Leaf nodes}

Note that, in this case, we have $V_t = \mathcal{D}(t)$, and hence for any state $\sigma = (H, \phi)$ a graph $\widetilde{G_t}$ satisfies the conditions of Definition \ref{def:gen-E} if and only if $H = \widetilde{G_t}$.  Thus 
$$\del_k(t,\sigma) = \left[ e(\widetilde{G_t}) - e(H)\right]_{\leq k}.$$

\subsubsection*{Introduce nodes}

Suppose that the introduce node $t$ has child $t'$, and that $\mathcal{D}(t) \setminus \mathcal{D}(t') = \{v\}$.  Given a state $\sigma \in \st(t)$, we define the \emph{introduce-inherited state} of $t'$ with respect to $\sigma$ to be the state $\sigma' = (H',\phi')$, where $H' = H \setminus v$ and for all $(F',\theta) \in X_{\mathcal{F},H'}$ we have $\phi'(F',\theta) = \phi(F',\theta)$.  It is clear that $\sigma'$ does indeed belong to $\st(t')$.

We make the following claim concerning the signatures of $t$ and $t'$; recall that $d_G(v)$ denotes the degree of the vertex $v$ in the graph $G$.

\begin{lma}
Let $t$ be an introduce node with child $t'$, and let $\sigma = (H,\phi) \in \st(t)$. Then
$$\del_k(t,\sigma) = \left[ \del_k(t',\sigma') + \left(d_{G[\mathcal{D}(t)]}(v) - d_H(v)\right) \right]_{\leq k},$$
where $\sigma'$ is the introduce-inherited state of $t'$ with respect to $\omega$.
\label{lma:gen-introduce}
\end{lma}

\begin{proof}
Observe first that the edges incident with $v$ in any set $E' \in \mathcal{E}(t,\sigma)$ must be precisely the set $E_v = \{uv \in G[V_t]: uv \notin E(H) \}$.  We shall in fact demonstrate that
$$\mathcal{E}(t,\sigma) = \left\lbrace E_v \cup E': E' \in  \mathcal{E}(t',\sigma')\right\rbrace,$$
and hence that
$$\mathcal{E}_k(t,\sigma) = \left\lbrace E_v \cup E': E' \in \mathcal{E}_k(t',\sigma'), \text{ and } |E_v \cup E'| \leq k \right\rbrace.$$
The result will then follow immediately by taking the minimum cardinality of an element in this set, using the observation that $E_v$ is disjoint from every $E' \in  \mathcal{E}_k(t',\sigma')$.

We first show that $\mathcal{E}(t,\sigma) \subseteq \{E_v \cup E': E' \in \mathcal{E}(t',\sigma') \}$.  Let $\widehat{E} \in \mathcal{E}(t,\sigma)$: we claim $\widehat{E} \setminus E_v \in \mathcal{E}(t',\sigma')$, where $\sigma' = (H',\phi')$ is the introduce-inherited state of $t'$ with respect to $\sigma$.  Set $\widetilde{G_t} = G[V_t] \setminus \widehat{E}$.  It suffices to check that $\widetilde{G_t}[V_{t'}]$ satisfies the three conditions of Definition \ref{def:gen-E}.  The first condition is trivially satisfied, by definition of $\sigma'$, and the second follows immediately from the fact that $\widehat{E} \in \mathcal{E}(t,\sigma)$ and hence $\widetilde{G_t}$ does not contain any $F \in \mathcal{F}$ as a subgraph.  For the third condition, suppose that $F'$ is an induced subgraph of some $F \in \mathcal{F}$ and that $\theta$ is an embedding of $F'$ into $\widetilde{G_t}[V_{t'}]$.  In this case $\theta$ is also an embedding of $F'$ into $\widetilde{G_t}$ and so (since $\widehat{E} \in \mathcal{E}(t,\sigma)$) we must have $\phi(F',\theta) = 1$; hence, by definition of $\sigma'$, we see that $\phi'(F',\theta) = 1$, as required.  Thus $\widetilde{G_t}[V_{t'}]$ satisfies the conditions of Definition \ref{def:gen-E} and we conclude that $\mathcal{E}(t,\sigma) \subseteq \{E_v \cup E': E' \in \mathcal{E}(t',\sigma') \}$, as required.

Conversely, we now show that $\{E_v \cup E': E' \in  \mathcal{E}(t',\sigma')\} \subseteq \mathcal{E}(t,\sigma)$.  Let $E' \in \mathcal{E}(t',\sigma')$, where $\sigma' = (H',\phi')$ is the introduce-inherited state of $t'$ with respect to $\sigma$;  we claim that $E_v \cup E' \in \mathcal{E}(t,\sigma)$.  Set $\widetilde{G_t}' = G[V_{t'}] \setminus E'$, and let $\widetilde{G_t} = G[V_t] \setminus (E' \cup E_v)$.  It suffices to demonstrate that $\widetilde{G_t}$ satisfies the conditions of Definition \ref{def:gen-E}.  

The first condition holds by definition.  For the second condition, suppose for a contradiction that $\widetilde{G_t}$ contains a copy of some $F \in \mathcal{F}$; let $\psi$ be an embedding of $F$ into $\widetilde{G_t}$.  Since $E' \in \mathcal{E}(t',\sigma')$, we know that $F$ is not a subgraph of $\widetilde{G_t}'$, so $\psi$ must map some vertex of $F$ to $v$.  It follows that $\psi|_{V_{t'}}$ is an embedding of some $F' \subset F$ into $\widetilde{G_t}'$ and hence by definition of $\mathcal{E}(t',\sigma')$ we must have $\phi'(F',\psi|_{V_{t'}}) = 1$.  By definition of $\sigma'$ we then have $\phi(F',\psi|_{V_t'}) = 1$ and hence, by condition 2(c) in the definition of a valid state, we would also have $\phi(F,\psi) = 1$ (since the properties of the tree decomposition ensure that every neighbour of $\psi^{-1}(v)$ in $F$ must belong to $\mathcal{D}(t')$).  It then follows that $\sigma$ is not a valid state for $t$ (by condition 2(d) of the definition of a valid state), giving the required contradiction.

For the third condition, suppose that $F \in \mathcal{F}$, $F'$ is an induced subgraph of $F$, and $\theta$ is an embedding of $F'$ into $\widetilde{G_t}$.  If $v \notin \theta(F')$ then $\theta$ is an embedding of $F'$ into $\widetilde{G_t}'$ so, since $E' \in \mathcal{E}(t',\sigma')$, we know that $\phi'(F',\theta) = 1$; by the definition of $\sigma'$ it follows that $\phi(F',\theta) = 1$, as required.  Thus we may assume that $v \in \theta(F')$; suppose that $u = \theta^{-1}(v)$.  In this case we see that $\theta|_{V(F') \setminus \{u\}}$ is an embedding of $F' \setminus u$ into $\widetilde{G_t}'$, so by definition of $\mathcal{E}(t',\sigma')$ we must have $\phi'(F' \setminus u, \theta|_{V(F')\setminus \{u\}}) = 1$ and hence by definition of $\sigma'$ we have $\phi(F' \setminus u, \theta|_{F'\setminus u}) = 1$.  Since we know $\sigma$ is a valid state for $t$, it then follows from condition 2(c) in the definition of a valid state that we must also have $\phi(F',\theta) = 1$, as required.

This completes the argument that $\{E_v \cup E': E' \in  \mathcal{E}(t',\sigma'), E_v \subseteq \widetilde{E_v}\} \subseteq \mathcal{E}(t,\sigma)$, and hence the proof.
\end{proof}

\subsubsection*{Forget nodes}

Suppose that the forget node $t$ has child $t'$, and that $\mathcal{D}(t') \setminus \mathcal{D}(t) = \{v\}$.  In this case, we need to consider the set of \emph{forget-inherited states} for $t'$.  Given $\sigma = (H,\phi) \in \st(t)$, we write $\sigma[t']_{\fgt}$ for the set of states $(H',\phi') \in \st(t')$ satisfying:
\begin{enumerate}
\item $H'[\mathcal{D}(t)] = H$, and
\item $\phi^{-1}(1) = \{(F',\theta|_{\mathcal{D}(t)}): (F',\theta) \in (\phi')^{-1}(1)\}$.
\end{enumerate}

We make the following claim concerning the signatures of $t$ and $t'$.

\begin{lma}
Let $t$ be a forget node with child $t'$, let $\sigma = (H,\phi) \in \st(t)$, and let $\sigma[t']_{\fgt}$ be the set of forget-inherited states of $t'$. Then
$$\del_k(t,\sigma) = \min_{\sigma' \in \sigma[t']_{\fgt}} \del_k(t',\sigma').$$
\label{lma:gen-forget}
\end{lma}
\begin{proof}
We shall in fact demonstrate that 
$$\mathcal{E}(t,\sigma) = \bigcup_{\sigma' \in \sigma[t']_{\fgt}} \left\lbrace \mathcal{E}(t',\sigma') \right\rbrace,$$
and hence that
$$\mathcal{E}_k(t,\sigma) = \bigcup_{\sigma' \in \sigma[t']_{\fgt}} \left\lbrace \mathcal{E}_k(t',\sigma') \right\rbrace,$$
from which the result will follow immediately by taking the minimum cardinality of an element in each set.

We begin by showing that $\bigcup_{\sigma' \in \sigma[t']_{\fgt}} \left\lbrace \mathcal{E}(t',\sigma') \right\rbrace \subseteq \mathcal{E}(t,\sigma)$.  Suppose that $E' \in \mathcal{E}(t',\sigma')$ for some $\sigma' \in \sigma[t']_{\fgt}$; we must verify that $E'$ also satisfies the three  conditions of Definition \ref{def:gen-E}.  Note that $V_t = V_{t'}$, and set $\widetilde{G_t} = G[V_t] \setminus E'$.  The first two conditions follow trivially from the definition of $\sigma[t']_{\fgt}$ and the fact that $E' \in \mathcal{E}(t',\sigma')$.  For the third condition, observe that if $\theta$ is an embedding of $F'$ into $\widetilde{G_t}$ then we must have $\phi'(F',\theta|_{\mathcal{D}(t')}) = 1$ and hence, by definition of $\sigma[t']_{\fgt}$ we also have $\phi(F',\theta|_{\mathcal{D}(t)}) = 1$, as required.

Conversely, we now argue that $\mathcal{E}(t,\sigma) \subseteq \bigcup_{\sigma' \in \sigma[t']_{\fgt}} \left\lbrace \mathcal{E}(t',\sigma') \right\rbrace$. Let $\widehat{E} \in \mathcal{E}(t,\sigma)$.  We claim that there exists $\sigma' \in \sigma[t']_{\fgt}$ such that $\widehat{E} \in \mathcal{E}(t',\sigma')$.  Set $\widetilde{G_t} = G[V_t] \setminus \widehat{E} = G[V_{t'}] \setminus \widehat{E}$, and let $\sigma' = (H',\phi')$ where $H' = \widetilde{G_t}[\mathcal{D}(t')]$ and, for every $(F',\theta) \in X_{\mathcal{F},H}$, we have $\phi'(F',\theta) = 1$ if and only if $\phi(F',\theta|_{\mathcal{D}(t)})=1$.  It is straightforward to verify that $(H',\phi') \in \sigma[t']_{\fgt}$.  We must verify that $\widetilde{G_t}$ satisfies the three conditions of Definition \ref{def:gen-E}.  The first condition is trivially satisfied by our choice of $\sigma'$, and the second follows from the fact that $\widehat{E} \in \mathcal{E}(t,\sigma)$.  For the third condition, suppose that $\theta$ is an embedding of $F'$ into $\widetilde{G_t}$; it follows that $\phi(F',\theta|_{\mathcal{D}(t)}) = 1$ and also that $\theta|_{\mathcal{D}(t')}$ is an embedding of $F'$ into $H'$ so, by definition of $\phi'$, we also have $\phi'(F',\theta|_{\mathcal{D}(t')})=1$, as required.  Thus we see that $\widehat{E} \in \mathcal{E}(t',\sigma')$, completing the argument that $\mathcal{E}(t,\sigma) \subseteq \bigcup_{\sigma' \in \sigma[t']_{\fgt}} \left\lbrace \mathcal{E}(t',\sigma') \right\rbrace$ and hence the proof.
\end{proof}

\subsubsection*{Join nodes}

Suppose that the join node $t$ has children $t_1$ and $t_2$.  Note that $\st(t_1) = \st(t_2) = \st(t)$.  In this case, we need to consider a set of pairs of inherited states for the two children $t_1$ and $t_2$.  We write $\st(t_1,t_2)$ for the Cartesian product $\st(t_1) \times \st(t_2)$ and, given any state $\sigma = (H,\phi) \in \st(t)$, we write $\sigma[t_1,t_2]_{\join}$ for the set of pairs of \emph{join-inherited states} $((H_1,\phi_1),(H_2,\phi_2)) \in \st(t_1,t_2)$ satisfying:
\begin{enumerate}
\item $H_1 = H_2 = H$, and
\item for every $(F',\theta) \in \phi^{-1}(0)$, and all induced subgraphs $F_1$ and $F_2$ of $F'$ such that $V(F_1) \cup V(F_2) = V(F')$ and $V(F_1) \cap V(F_2) = \dom(\theta)$, either $\phi_1(F_1,\theta|_{V(F_1)})=0$ or $\phi_2(F_2,\theta|_{V(F_2)})=0$.
\end{enumerate}

We now make the following claim about the signature of $t$ and those of $t_1$ and $t_2$.
\begin{lma}
Let $t$ be a join node with children $t_1$ and $t_2$, and let $\sigma = (H,\phi) \in \st(t)$.  Then
\begin{align*}
\del_k(t,\sigma) & = \left[ \min_{\substack{(\sigma_1,\sigma_2) \in \\ \sigma[t_1,t_2]_{\join}}} \{\del_k(t_1,\sigma_1) + \del_k(t_2,\sigma_2) - \left(e(G[\mathcal{D}(t)]) - e(H)\right) \right]_{\leq k}.
\end{align*}
\label{lma:gen-join}
\end{lma}
\begin{proof}
We will demonstrate that 
$$\mathcal{E}(t,\sigma) = \bigcup_{(\sigma_1,\sigma_2) \in \sigma[t_1,t_2]_{\join}} \lbrace E_1 \cup E_2: E_1 \in \mathcal{E}(t_1,\sigma_1), E_2 \in \mathcal{E}(t_2,\sigma_2) \rbrace,$$
and hence that
$$\mathcal{E}_k(t,\sigma) = \bigcup_{(\sigma_1,\sigma_2) \in \sigma[t_1,t_2]_{\join}} \lbrace E_1 \cup E_2: E_1 \in \mathcal{E}_k(t_1,\sigma_1), E_2 \in \mathcal{E}_k(t_2,\sigma_2), |E_1 \cup E_2| \leq k \rbrace.$$
Notice that, for any $E_1 \in \mathcal{E}(t_1,\sigma_1)$ and $E_2 \in \mathcal{E}(t_2,\sigma_2)$ the properties of the tree decomposition ensure that $E_1 \cap E_2 = E(G[\mathcal{D}(t)]) \setminus E(H)$, so the result will follow immediately by taking the minimum cardinality of an element from each set.

We first show that $\mathcal{E}(t,\sigma) \supseteq \bigcup_{(\sigma_1,\sigma_2) \in \sigma[t_1,t_2]_{\join}} \lbrace E_1 \cup E_2: E_1 \in \mathcal{E}(t_1,\sigma_1), E_2 \in \mathcal{E}(t_2,\sigma_2) \rbrace.$  Let $(\sigma_1,\sigma_2) \in \sigma[t_1,t_2]_{\join}$, and suppose that $E_1 \in \mathcal{E}(t_1,\sigma_1)$ and $E_2 \in \mathcal{E}(t_2,\sigma_2)$; we claim that $E_1 \cup E_2 \in \mathcal{E}(t,\sigma)$.  It suffices to show that $\widetilde{G_t} = G[V_t] \setminus (E_1 \cup E_2)$ satisfies the three conditions of Definition \ref{def:gen-E}.

The first condition is trivially satisfied.  For the second and third conditions suppose that, for some $F \in \mathcal{F}$, $F'$ is an induced subgraph of $F$ and $\theta$ is an embedding of $F'$ into $\widetilde{G_t}$.  Write $F_i'$ for $F'[V_{t_i}]$ (for $i \in \{1,2\}$) and note that $\theta|_{V_{t_i}}$ is an embedding of $F_i'$ into $G[V_{t_i}]$; thus, as $E_i \in \mathcal{E}(t_i,\sigma_i)$ we must have $\phi_i(F_i',\theta|_{V_{t_i}}) = 1$.  Since $V(F') = V(F_1') \cup V(F_2')$, and $\theta^{-1}(\theta(V(F')) \cap \mathcal{D}(t)) = V(F_1') \cap V(F_2')$, it follows from the definition of $\sigma[t_1,t_2]_{\join}$ that $\phi(F',\theta|_{\mathcal{D}(t)}) = 1$. This demonstrates that the third condition is satisfied.  To see that the second condition is also satisfied, observe that if $F' = F$ then we would by the reasoning above have $\phi(F,\theta|_{\mathcal{D}(t)})=1$, contradicting the fact that $\sigma$ is a valid state, so we see that the second condition must also be satisfied.  Thus we see that $E_1 \cup E_2 \in \mathcal{E}(t,\sigma)$, as required.

Conversely, we now show that $\mathcal{E}(t,\sigma) \subseteq \bigcup_{(\sigma_1,\sigma_2) \in \sigma[t_1,t_2]_{\join}} \lbrace E_1 \cup E_2: E_1 \in \mathcal{E}(t_1,\sigma_1), E_2 \in \mathcal{E}(t_2,\sigma_2) \rbrace$.  Let $\widehat{E} \in \mathcal{E}(t,\sigma)$, and set $\widetilde{G_t} = G[V_t] \setminus \widehat{E}$; we also let $E_i = \widehat{E} \cap E(G[V_{t_i}])$ for $i \in \{1,2\}$.  Now set $\sigma_i = (H,\phi_i)$ where $\phi_i(F',\theta) = 1$ if and only if $\theta$ can be extended to an embedding of $F'$ into $\widetilde{G_t}[V_{t_i}]$.  

To verify that $(\sigma_1,\sigma_2) \in \sigma[t_1,t_2]_{\join}$ suppose, for a contradiction, that there exists $(F',\theta) \in \phi^{-1}(0)$, and induced subgraphs $F_1$ and $F_2$ of $F'$ such that $V(F_1) \cup V(F_2) = V(F')$ and $V(F_1) \cap V(F_2) = \dom(\theta)$, and $\phi_1(F_1,\theta|_{V(F_1)}) = \phi_2(F_2,\theta|_{V(F_2)})=1$.  By definition of $\phi_1, \phi_2$ it then follows (for $i \in \{1,2\}$) that $\theta|_{V(F_i)}$ can be extended to an embedding $\theta_i$ of $F_i$ into $\widetilde{G_t}[V_{t_i}]$; note that $\theta_1$ and $\theta_2$ agree on $V(F_1) \cap V(F_2)$ and hence on all vertices of $\theta^{-1}(\mathcal{D}(t))$.  We can therefore define $\widehat{\theta}$ to be the injective function which agrees with both $\theta_1$ and $\theta_2$ on their respective domains.  It is clear that $\widehat{\theta}$ defines an embedding of $F'$ into $\widetilde{G_t}$ and so (as we are assuming that $\widetilde{G_t}$ satisfies the conditions of Definition \ref{def:gen-E}) we must have $\phi(F',\widehat{\theta}|_{\mathcal{D}(t)})=1$; however, by definition, $\widehat{\theta}|_{\mathcal{D}(t)} = \theta$, so we have  $\phi(F',\theta)=1$, contradicting our initial assumption.  We may therefore conclude that $(\sigma_1,\sigma_2) \in \sigma[t_1,t_2]_{\join}$.

We now proceed to show that $E_i \in \mathcal{E}(t_i,\sigma_i)$ for each $i \in \{1,2\}$.  It suffices to demonstrate that each $G[V_{t_i}] \setminus E_i = \widetilde{G_t}[V_{t_i}]$ satisfies the conditions of Definition \ref{def:gen-E}.  The first two conditions are trivially satisfied.  For the third condition, suppose that $F'$ is an induced subgraph of some $F \in \mathcal{F}$ and that $\theta$ is an embedding of $F'$ into $\widetilde{G_t}[V_{t_i}]$; it follows immediately from the definition of $\phi_i$ that $\phi_i(F',\theta|_{\mathcal{D}(t_1)}) = 1$, as required.

This completes the argument that $\bigcup_{(\sigma_1,\sigma_2) \in \sigma[t_1,t_2]_{\join}} \lbrace E_1 \cup E_2: E_1 \in \mathcal{E}(t_1,\sigma_1), E_2 \in \mathcal{E}(t_2,\sigma_2) \rbrace \subseteq \mathcal{E}(t,\sigma)$, and hence the proof.
\end{proof}

\subsection{Running time and extensions}
\label{sec:gen-runtime}

Observe that we can precompute the list of all valid states for a given node $t$ of the tree decomposition, as this depends only on the subgraph $G[\mathcal{D}(t)]$, for which there are at most $2^{\binom{w+1}{2}} < 2^{w^2}$ possibilities.  To find the set of valid states for a bag which induces a given subgraph in $G$, we first determine all possibilities for $H$: this must be a spanning subgraph of $G[\mathcal{D}(t)]$, and there are most $2^{w^2}$ possibilities.  Next, we determine the set $X_{\mathcal{F},H}$ of pairs $(F',\theta)$ such that $F'$ is an induced subgraph of some $F \in \mathcal{F}$ and $\theta$ is an embedding into $H$ of some induced subgraph $F''$ of $F'$; to do this, we consider each of the at most $|\mathcal{F}| (w+3)^r$ possible pairs $(F',\theta)$ (as described in Section \ref{sec:gen-signature}), and verify in time at most $O(w^2)$ for each such possibility whether $\theta$ does indeed define an embedding.  Finally, for each possible function $\phi: X_{\mathcal{F},H} \rightarrow \{0,1\}$ (of which there are at most $2^{|\mathcal{F}| (w+3)^r}$), we must determine whether conditions 2(a)-(d) in the definition of a valid state are satisfied: conditions (a) and (d) can each be checked in constant time, whereas (b) requires time at most $O(2^r)$ (to consider every possible induced subgraph of some $F_1 \in \mathcal{F}$, which can have at most $r$ vertices) and (c) requires time at most $O(wr^2)$ (to consider every possible choice of a pair of vertices $u \in V(F) \setminus V(F_1)$ and $v \in V(H) \setminus \image(\theta)$, and to verify that extending the mapping preserves all edges in $F$ that are incident with $u$).  Thus we can generate a lookup table of the valid states for any given node in time 
$$O \left(2^{w^2} \cdot 2^{w^2} \left(|\mathcal{F}|(w+3)^r \cdot w^2 + 2^{|\mathcal{F}| (w+3)^r} \left(2^r + wr^2\right)\right)\right) = 2^{O \left(|\mathcal{F}|w^r\right)}.$$

Moreover, we can also precompute, for each possible state $\sigma = (H, \phi)$, the corresponding introduce-inherited state and sets of forget-inherited and join-inherited states.  

For an introduce node, there is a single introduce-inherited state for the child node, and this can clearly be computed in time $O(|X_{\mathcal{F},H'}|) = O(|\mathcal{F}(w+3)^r)$.

For a forget node, we need to consider all valid states $\sigma' = (H', \phi')$ for the child.  We can first verify that $H$ and $H'$ have the correct relationship (taking time at most $O(w^2)$), then we compute the set $\{(F',\theta|_{\mathcal{D}(t)}):(F',\theta) \in (\phi')^{-1}(1)\}$ (taking time at most $O(|X_{\mathcal{F},H'}|) = O(|\mathcal{F}|(w+3)^r)$) and verify (in time $O(|X_{\mathcal{F},H}|)$) that this is equal to $\phi^{-1}(1)$.  Thus the total time required to compute the set of forget-inherited states with respect to $\sigma$ is $O(2^{w^2 + |\mathcal{F}| (w+3)^r} w^2 |\mathcal{F}|(w+3)^r) = 2^{O(|\mathcal{F}|w^r)}$.

For a join node, we need to consider all pairs of valid states $(\sigma_1,\sigma_2)$ for the children, where $\sigma_i = (H_i,\phi_i)$.  We can first verify that $H_1 = H_2 = H$ in time $O(w^2)$.  For the condition on $\phi$, $\phi_1$ and $\phi_2$, we consider each $(F',\theta) \in X_{\mathcal{F},H}$, determine all possibilities for $F_1$ and $F_2$ so that $V(F_1) \cup V(F_2) = V(F)$ (of which there are at most $3^r$, as each vertex of $F'$ belongs to either $F_1$, $F_2$ or both) and then verify in constant time whether the condition is satisfied for this choice of $(F',\theta)$ and $(F_1,F_2)$.  Thus the total time required to compute the set of join-inherited states with respect to $\sigma$ is $O((2^{w^2 + |\mathcal{F}| (w+3)^r})^2 w^2 |\mathcal{F}|(w+3)^r 3^r) = 2^{O(|\mathcal{F}|w^r)}$.

Hence we can perform all precomputation in time $2^{O(|\mathcal{F}|w^r)}$.  Having precomputed the sets of states we then, at each of the $O(n)$ nodes of the nice tree decomposition, iterate over all of the valid states for the node, of which there are at most $2^{w^2 + |\mathcal{F}|(w+3)^r}$.  For each possible state of a given node, we will need to consider a collection of at most $2^{2(w^2 + |\mathcal{F}|(w+3)^r)}$ inherited states (or pairs of states, in the case of a join node) of the child node(s), and to perform various constant-time operations for each such state.  At each of the $O(n)$ nodes we therefore do
$$O\left(2^{3(w^2 + |\mathcal{F}|(w+3)^r)}\right) = 2^{O(|\mathcal{F}|w^r)}$$
work, so this phase of the computation requires time $2^{O(|\mathcal{F}|w^r)}n$, and the overall time complexity of the algorithm is therefore $2^{O(|\mathcal{F}|w^r)}n$.  We can bound the size of $\mathcal{F}$ by $2^{2^r}$, to obtain an upper bound on the running time that depends only on $n$, $w$ and $r$.

It is straightforward to adapt the algorithm to solve the related problem in which we wish to delete edges to that the resulting graph contains no \emph{induced} copy of any $F \in \mathcal{F}$: we simply consider strong embeddings instead of embeddings, and make analogous changes to the definition of a valid state.  This does not change the running time as a function of $n$, $w$ and $|\mathcal{F}|$, but for monotone properties we would typically need to consider a larger set of forbidden \emph{induced} subgraphs than forbidden subgraphs, resulting in an increased running time.

For simplicity, we have only described the most basic version of the algorithm.  However, it is straightforward to extend it to deal with a more realistic situation in which the cost of deleting different edges may vary: in this case, given a cost function $f:E(G) \rightarrow \mathbb{N}$, we instead define $\del_k(t,\sigma)$ to be $\min_{E' \in \mathcal{E}(\sigma,t)} \sum_{e \in E'} f(e)$.  Additionally, if we wish to output an optimal set of edges to delete, we can simply record, for each node $t$ and each state $\sigma \in \st(t)$, a set of edges $E' \in \mathcal{E}_k(t,\sigma)$ such that $|E'| = \del_k(t,\sigma)$ (note that in general there may be many such optimal sets); computing such a set from the relevant sets for the node's children requires only basic set operations.  An element of $\mathcal{E}(t_r,\sigma)$, where $t_r$ is the root of the tree decomposition and $\del_k(t_r,\sigma) = \min_{\sigma \in \st(t_r)}$ is then an optimal solution for the problem.  Neither of these adaptations changes the asymptotic running time of the algorithm.

\section{A specialised algorithm for \compprobname}
\label{sec:cpt-alg}

In this section, we turn our attention to the special case of \compprobname ~and describe a more efficient algorithm for this important special case of the problem.  Specifically, we prove the following theorem.

\begin{thm}
There exists an algorithm to solve \compprobname ~in time $O((wh)^{2w}n)$ on an input graph with $n$ vertices whose treewidth is at most $w$.
\end{thm}

The general strategy is very similar to that employed in our general algorithm, but in this special case we can determine whether our input graph is a yes- or no-instance without using all of the information contained in the signature of the root of the tree indexing the decomposition.  Instead we can record an abridged \emph{component-signature} for each node, reducing the overall running time of the algorithm.  

This relies on two key observations.  Firstly, an optimal solution will never delete an edge $uv$ if there remains another path from $u$ to $v$, so it suffices to record which vertices in a bag are permitted to belong to the same component, rather than the exact subgraph induced by the bag.  Secondly, we do not need to know precisely which partial forbidden subgraphs intersect the bag and how, rather just the size of the component that contains each vertex of the bag.

We define the component-signature of a node in Section \ref{sec:cpt-signature}, before describing mathematically how we compute the component-signature of a bag indexed by a given node from the signatures of its children in Section \ref{sec:cpt-recursive}, and discussing the running time and a number of extensions in Section \ref{sec:cpt-runtime}.  As this algorithm achieves a faster running time and answers a question of particular interest from the point of view of our epidemiological application, the specialised algorithm is likely to be of more practical use than the general algorithm described in Section \ref{sec:gen-alg}; we express the procedure in pseudocode in Section \ref{sec:cpt-algorithmic}, and give some initial experimental results on its application to cattle trading networks in Section \ref{sec:expt}.

\subsection{The component-signature of a node}
\label{sec:cpt-signature}

In this section, we describe the information we compute for each node, and define the \emph{component-signature} of a node.

Throughout the algorithm, we need to record the possible component-states corresponding to a given bag.  A valid \emph{component-state} of a bag $\mathcal{D}(t)$ is a pair consisting of:
\begin{enumerate}
\item a partition $\mathcal{P}$ of $\mathcal{D}(t)$ into disjoint, non-empty subsets or \emph{blocks} of size at most $h$, and
\item a function $c: \mathcal{P} \rightarrow [h]$ such that, for each $X \in \mathcal{P}$, $|X| \leq c(X)$.
\end{enumerate}
We will write $u \sim_\mathcal{P} v$ to indicate that $u$ and $v$ belong to the same block of $\mathcal{P}$.

Intuitively, $\mathcal{P}$ tells us which vertices are allowed to belong to the same component of the graph we obtain after deleting edges and $c$ tells us the maximum number of vertices which are permitted in components corresponding to a given block of the partition.

For any bag $\mathcal{D}(t)$, we denote by $\cst(t)$ the set of possible component-states of $\mathcal{D}(t)$.  Note that there are at most $B_w$ partitions of a set of size $w$ (where $B_w$ is the $w^{th}$ Bell number) and at most $h^w$ functions from a set of size at most $w$ to $[h]$; thus the total number of valid component-states for $\mathcal{D}(t)$ is at most $B_wh^w < (wh)^w$ (although not all possible combinations of a partition and a function will give rise to a valid component-state).

Given any component-state, we make the following definition.
\begin{adef}
\label{def:comp-E}
Let $\sigma = (\mathcal{P},c) \in \cst(t)$.  Then $\mathcal{E}^c(t,\sigma)$ is the set of edge-sets $E' \subset E(G[V_t])$ such that $\widetilde{G_t} = G[V_t] \setminus E'$ has the following properties:
\begin{enumerate}
\item for each connected component $C$ of $\widetilde{G_t}$:
\begin{enumerate}
\item $|V(C)| \leq h$, and
\item if $C_t = V(C) \cap \mathcal{D}(t) \neq \emptyset$, then $C_t$ is contained in a single block $X_C$ of $\mathcal{P}$, 
\end{enumerate}
\item for each block $X$ in $\mathcal{P}$, the total number of vertices in connected components of $\widetilde{G_t}$ that intersect $X$ is at most $c(X)$.
\end{enumerate}
\end{adef}
Note that, whenever $\sigma$ is a valid component-state for $t$, the set $\mathcal{E}^c(t,\sigma)$ will be non-empty: setting $E' = E(G[V_t])$ will always satisfy both conditions.  Since we are interested in determining whether it is possible to delete at most $k$ edges to obtain a graph with maximum component size $h$, we will (as in Section \ref{sec:gen-alg}) primarily be interested in a subset of $\mathcal{E}^c(t,\sigma)$: for any node $t$ and $\sigma \in \cst(t)$ we define this subset as
$$\mathcal{E}_k^c(t,\sigma) = \{E' \in \mathcal{E}^c(t,\sigma): |E'| \leq k\}.$$
We then define
$$\cdel_k(t,\sigma) = \min_{E' \in \mathcal{E}_k^c(t,\sigma)} |E'|.$$
Finally, we define the \emph{component-signature} of a node $t$ to be the function $\csig_t: \cst(t) \rightarrow \{0,1,\ldots,k,\infty\}$ such that $\csig_t(\sigma) = \cdel_k(t,\sigma)$.  

Just as in Section \ref{sec:gen-alg}, our input graph is a yes-instance to \compprobname ~if and only if there exists some $\sigma \in \cst(t_r)$ such that $\csig_{t_r}(\sigma) \leq k$, where $t_r$ is the root of the tree indexing the decomposition.

\subsection{Computing component-signatures recursively}
\label{sec:cpt-recursive}

In this section we describe how the component-signature of any given node in a nice tree decomposition can be determined from the component-signatures of its children (if any).  As before, this relationship depends on the type of node under consideration, and we discuss each type of node in turn, describing how to compute $\cdel_k(t,\sigma)$ for an arbitrary component-state $\sigma$.

\subsubsection*{Leaf nodes}

Note that, in this case, we have $V_t = \mathcal{D}(t)$, and hence for any component-state $\sigma = (\mathcal{P}, c)$ a graph $\widetilde{G_t}$ satisfies the conditions of Definition \ref{def:comp-E} if and only if the only edges in $\widetilde{G_t}$ are between vertices that belong to the same block of $\mathcal{P}$.  Thus
$$\cdel_k(t,\sigma) = \left[ \left| \{uv \in E \left( G[\mathcal{D}(t)] \right) : u \not\sim_{\mathcal{P}} v\}\right| \right]_{\leq k}.$$

\subsubsection*{Introduce nodes}

Suppose that the introduce node $t$ has child $t'$, and that $\mathcal{D}(t) \setminus \mathcal{D}(t') = \{v\}$.  Given a component-state $\sigma = (\mathcal{P},c) \in \cst(t)$, where $\mathcal{P} = \{X_1,\ldots,X_r\}$ and $v \in X_r$, we denote by $\sigma[t']_{\cintr}$ the set of \emph{introduce-inherited component-states} of $t'$.  A component-state $(\mathcal{P}',c')$ belongs to $\sigma[t']_{\cintr}$ if and only if it is a valid component-state for $\mathcal{D}(t')$ which additionally satisfies the following conditions:
\begin{enumerate}
\item $\mathcal{P}' = \{X_1,\ldots,X_{r-1}, Y_{1},\ldots,Y_s\}$, with $s \geq 1$ and $Y_{1} \cup \cdots \cup Y_s = X_r \setminus \{v\}$,
\item $c'(X_i) = c(X_i)$ for $1 \leq i \leq r-1$, and $\sum_{i=1}^s c'(Y_i) = c(X_r) - 1$.
\end{enumerate}

We make the following claim concerning the component-signatures of $t$ and $t'$.

\begin{lma}
Let $t$ be an introduce node with child $t'$, and let $\sigma = (H,\mathcal{P},c) \in \cst(t)$. Then
$$\cdel_k(t,\sigma) = \left[ \min_{\sigma' \in \sigma[t']_{\cintr}} \cdel_k(t',\sigma') + |\{uv \in G[V_t]: u \not\sim_{\mathcal{P}} v\}| \right]_{\leq k}.$$
\label{lma:introduce}
\end{lma}
\begin{proof}
Observe first that the edges incident with $v$ in any set $E' \in \mathcal{E}^c(t,\sigma)$ must contain the set $E_v = \{uv \in G[V_t]: u \not\sim_{\mathcal{P}} v\}$.  As in the proof of Lemma \ref{lma:gen-introduce}, it suffices to prove that 
$$\mathcal{E}^c(t,\sigma) = \left\lbrace \widetilde{E_v} \cup E': E' \in \bigcup_{\sigma' \in \sigma[t']_{\cintr}} \mathcal{E}^c(t',\sigma'), E_v \subseteq \widetilde{E_v} \right\rbrace.$$

We first show that $\mathcal{E}^c(t,\sigma) \subseteq \{\widetilde{E_v} \cup E': E' \in \bigcup_{\sigma' \in \sigma[t']_{\cintr}} \mathcal{E}^c(t',\sigma'), E_v \subseteq \widetilde{E_v} \}.$  Let $\widehat{E} \in \mathcal{E}^c(t,\sigma)$: we claim that there exists $\sigma' \in \sigma[t']_{\cintr}$ such that $\widehat{E} \setminus v \in \mathcal{E}^c(t',\sigma')$.

Set $\widetilde{G_t} = G[V_t] \setminus \widehat{E}$, and let $C$ be the component of $\widetilde{G_t}$ containing $v$.  Suppose that $\mathcal{P} = X_1,\ldots,X_r$, with $v \in X_r$, and that $C \setminus v$ has connected components $C_1,\ldots,C_{\ell}$.  We now define $\mathcal{P}'$ to be the partition $\{Y_1,\ldots,Y_{r+\ell-1}\}$ where $Y_i = X_i$ for $1 \leq i \leq r-1$, and $Y_{r+j-1} = V(C_j) \cap \mathcal{D}(t')$ for $1 \leq j \leq \ell$.  We also define $c':\mathcal{P}' \rightarrow [h]$ by setting
\begin{equation*}
c(Y_i) = \begin{cases}
			c(X_i)	& \text{if $1 \leq i \leq r-1$} \\
			|V(C_{i-r+1})| & \text{if $r \leq i \leq r + \ell - 1,$}
	   \end{cases}
\end{equation*}
and then set $\sigma' = (\mathcal{P}',c')$.  It is straightforward to verify that $\sigma' \in \sigma[t']_{\cintr}$.  We now claim that $\widehat{E} \setminus v \in \mathcal{E}^c(t',\sigma')$.  To prove this claim, we set $\widetilde{G_t}' = G[V_{t'}] \setminus (\widehat{E} \setminus v)$; we need to verify that $\widetilde{G_t}'$ satisfies the conditions of Definition \ref{def:comp-E}.  Condition 1(a) is trivially satisfied, so we consider the remaining conditions.

For condition 1(b), suppose that $C$ is a connected component of $\widetilde{G_t}'$ such that $C_t = V(C) \cap \mathcal{D}(t') \neq \emptyset$.  Since $C$ is contained in some connected component $\widehat{C}$ of $\widetilde{G_t}$, and we know that every connected component of $\widetilde{G_t}$ is contained in a single block of $\mathcal{P}$, it follows that $C$ must be contained in a single block $X_i$ of $\mathcal{P}$.  There are now two cases to consider, depending on whether $i = r$.  If we have $i \neq r$ (implying that $v \notin \widehat{C}$, so $\widehat{C} = C$), it follows from the definition of $\mathcal{P}'$ that $C$ is contained in the block $Y_i$ of $\mathcal{P}'$.  If, on the other hand, we have $i=r$, it follows from the definition of $\mathcal{P}'$ that there is some $j \in \{1,\ldots,\ell\}$ such that $Y_{r+j-1} = C$.  Thus the condition is satisfied in either case.

For condition 2, consider a block $Y$ of $\mathcal{P}'$.  If $Y$ is also a block of $\mathcal{P}$, then (as $\widetilde{G_t}'$ is a subgraph of $\widetilde{G_t}$) this condition is trivially true, so suppose that this is not the case.  Then, by construction of $\mathcal{P}'$, $Y = V(C') \cap \mathcal{D}(t')$ for some (maximal) connected component $C'$ of $\widetilde{G_t}'$, and we have $c'(Y) = |V(C')|$, so condition 2 is indeed satisfied.

This completes the argument that $\widetilde{G_t}'$ satisfies the conditions of Definition \ref{def:comp-E}, implying that $\widehat{E} \setminus v \in \mathcal{E}^c(t',\sigma')$ and hence $\widehat{E} \in \{\widetilde{E_v} \cup E': E' \in \bigcup_{\sigma' \in \sigma[t']_{\cintr}} \mathcal{E}^c(t',\sigma'), E_v \subseteq \widetilde{E_v}\}$, as required.

Conversely, we now show that $\{\widetilde{E_v} \cup E': E' \in \bigcup_{\sigma' \in \sigma[t']_{\cintr}} \mathcal{E}^c(t',\sigma'), E_v \subseteq \widetilde{E_v}\} \subseteq \mathcal{E}^c(t,\sigma)$.  Let $\widetilde{E_v} \supseteq E_v$, and let $E' \in \mathcal{E}^c(t',\sigma')$ for some $\sigma' = (\mathcal{P}',c') \in \sigma[t']_{\cintr}$ (where $\mathcal{P}' = \{Y_1,\ldots,Y_p\}$, with $Y_i = X_i$ for $1 \leq i \leq r-1$ and $Y_r \cup \cdots \cup Y_p = X_r \setminus \{v\}$).  We claim that $\widetilde{E_v} \cup E' \in \mathcal{E}^c(t,\sigma)$.

Set $\widetilde{G_t}' = G[V_{t'}] \setminus E'$, and let $\widetilde{G_t} = G[V_t] \setminus (E' \cup E_v)$.  It suffices to demonstrate that $\widetilde{G_t}$ satisfies the conditions of Definition \ref{def:comp-E}.  

For the first condition, let $C$ be a (maximal) connected component of $\widetilde{G_t}$.  If $v \notin C$, then $C$ is a connected component of $\widetilde{G_t}$ so we know that $|V(C)| \leq h$, and moreover that if $C_t = V(C) \cap \mathcal{D}(t) \neq \emptyset$ (note that $V(C) \cap \mathcal{D}(t) = V(C)\cap \mathcal{D}(t')$ as $v \notin C$) then $C_t$ is contained in a single block $X_c$ of $\mathcal{P}'$ and hence a single block of $\mathcal{P}$ (as $\mathcal{P}'$ refines $\mathcal{P}$).  So it remains to consider the case that $v \in C$.  Let $C_1,\ldots,C_{\ell}$ be the (maximal) connected components of $C \setminus v$.  Note that it follows from the properties of a tree decomposition that $C_i \cap \mathcal{D}(t') \neq \emptyset$ for each $i$, as there must be an edge from $v$ to at least one vertex in each $C_i$; suppose that $u_i \in C_i \cap \mathcal{D}(t')$ for each $i$.  Then $u_1,\ldots,u_{\ell}$ all belong to the same component of $H$ as $v$, so as $\sigma$ is a valid component-state for $t$ we must have $u_1,\ldots,u_{\ell} \in X_r$.  We know that each $C_i$ is a connected component of $\widetilde{G_t}'$ and so $V(C_i) \cap \mathcal{D}(t)$ must be contained in a single block of $\mathcal{P}'$; since $u_i \in X_r$, it must be that $V(C_i) \cap \mathcal{D}(t') \subseteq Y_j$ for some $r \leq j \leq p$.  Thus
$$V(C) \cap \mathcal{D}(t) = \bigcup_{1 \leq i \leq \ell} (C_i \cap \mathcal{D}(t')) \cup \{v\} \subseteq \bigcup_{r \leq j \leq s} Y_j \cup \{v\} = X_r,$$
so $V(C) \cap \mathcal{D}(t)$ is contained in a single block of $\mathcal{P}$, as required.

For the second condition, let $X$ be a block of $\mathcal{P}$.  If $v \notin X$, then $X$ is also a block of $\mathcal{P}'$, and moreover the vertices belonging to components of $\widetilde{G_t}$ that intersect $X$ are the same as those belonging to components of $\widetilde{G_t}'$ that intersect $X$, so we are done in this case.  Thus we may assume that $v \in X$, and that $X \setminus \{v\} = Y_r \cup \cdots \cup Y_p$.  Then the number of vertices belonging to components of $\widetilde{G_t}$ that intersect $X$ is exactly one more than the number of vertices belonging to components of $\widetilde{G_t}'$ that intersect $Y_r \cup \cdots \cup Y_s$.  Thus, the number of vertices belonging to components of $\widetilde{G_t}$ that intersect $X$ is at most
$$1 + \sum_{i=r}^s c(Y_i) = 1 + c(X) - 1 = c(X), $$
as required.

This completes the argument that $\{E_v \cup E': E' \in \bigcup_{\sigma' \in \sigma[t']_{\cintr}} \mathcal{E}^c(t',\sigma'), E_v \subseteq \widetilde{E_v}\} \subseteq \mathcal{E}^c(t,\sigma)$, and hence the proof.
\end{proof}

\subsubsection*{Forget nodes}

Suppose that the forget node $t$ has child $t'$, and that $\mathcal{D}(t') \setminus \mathcal{D}(t) = \{v\}$.  In this case, we need to consider the set of \emph{forget-inherited component-states} for $t'$.  Given $\sigma = (\mathcal{P},c) \in \cst(t)$, we write $\sigma[t']_{\cfgt}$ for the set of component-states $(\mathcal{P}',c') \in \cst(t')$ satisfying:
\begin{enumerate}
\item $\mathcal{P} = \mathcal{P}' \setminus \{v\}$, and 
\item for any block $Y$ in $\mathcal{P}$, if $Y \setminus \{v\} \neq \emptyset$, we have $c'(Y) = c(Y \setminus \{v\})$.
\end{enumerate}

We make the following claim concerning the component-signatures of $t$ and $t'$.

\begin{lma}
Let $t$ be a forget node with child $t'$, let $\sigma = (\mathcal{P},c) \in \cst(t)$, and let $\sigma[t']_{\cfgt}$ be the set of inherited component-states of $t'$. Then
$$\cdel_k(t,\sigma) = \min_{\sigma' \in \sigma[t']_{\cfgt}} \cdel_k(t',\sigma').$$
\label{lma:forget}
\end{lma}
\begin{proof}
As in the proof of Lemma \ref{lma:gen-forget}, it suffices to prove that 
$$\mathcal{E}^c(t,\sigma) = \bigcup_{\sigma' \in \sigma[t']_{\cfgt}} \left\lbrace \mathcal{E}^c(t',\sigma') \right\rbrace.$$
It is straightforward to verify that, whenever $E' \in \mathcal{E}^c(t',\sigma')$ for some $\sigma' \in \sigma[t']_{\cfgt}$, we also have $E' \in \mathcal{E}^c(t,\sigma)$, implying that $\bigcup_{\sigma' \in \sigma[t']_{\cfgt}} \left\lbrace \mathcal{E}^c(t',\sigma') \right\rbrace \subseteq \mathcal{E}^c(t,\sigma)$.  We now argue that we also have $\mathcal{E}^c(t,\sigma) \subseteq \bigcup_{\sigma' \in \sigma[t']_{\cfgt}} \left\lbrace \mathcal{E}^c(t',\sigma') \right\rbrace$.  Let $\widehat{E} \in \mathcal{E}^c(t,\sigma)$.  We claim that there exists $\sigma' \in \sigma[t']_{\cfgt}$ such that $\widehat{E} \in \mathcal{E}^c(t',\sigma')$.

Set $\widetilde{G_t} = G[V_t] \setminus \widehat{E} = G[V_{t'}] \setminus \widehat{E}$.  Let $C$ be the (maximal) connected component of $\widetilde{G_t}$ that contains $v$.  If $C \cap \mathcal{D}(t) = \emptyset$, we set $\mathcal{P}'$ to be the partition obtained from $\mathcal{P}$ by adding one additional block containing only $v$; if $C \cap \mathcal{D}(t) \neq \emptyset$ then this set of vertices must be contained in a single block $X$ of $\mathcal{P}$, and we set $\mathcal{P}'$ to be the partition obtained from $\mathcal{P}$ by adding $v$ to the block $X$.   

In the case $C \cap \mathcal{D}(t) \neq \emptyset$, there is a unique choice of $c':\mathcal{P}' \rightarrow [h]$ which will satisfy the definition of $\sigma[t']_{\cfgt}$: we set $c'(X') = c(X')$ for $X' \neq X$, and $c'(X \cup \{v\}) = c(X)$.  If $C \cap \mathcal{D}(t) = \emptyset$, we define for any $Y \in \mathcal{P}'$
\begin{equation*}
c'(Y) = \begin{cases}
			c(Y) & \text{if $Y \in \mathcal{P}$} \\
			|V(C)|  & \text{if $Y = \{v\}$.}
	   \end{cases}
\end{equation*}

We now set $\sigma' = (\mathcal{P}',c')$.  It is straightforward to verify that $\sigma' \in \sigma[t']_{\cfgt}$.  It remains to demonstrate that $\widehat{E} \in \mathcal{E}^c(t',\sigma')$; to do so we will argue that $\widetilde{G_t}$ satisfies all the conditions of Definition \ref{def:comp-E}.  Condition 1(a) is trivially satisfied.

For condition 1(b), let $C$ be a (maximal) connected component of $\widetilde{G_t}$, and suppose that $C_t'=V(C) \cap \mathcal{D}(t') \neq \emptyset$.  If $C_t' = \{v\}$ then this condition is trivially satisfied, so we may assume that $V(C) \cap \mathcal{D}(t) \neq \emptyset$.  It then follows from the fact that $\widehat{E} \in \mathcal{E}^c(t,\sigma)$ that $V(C) \cap \mathcal{D}(t)$ is contained in a single block of $\mathcal{P}$ and hence of $\mathcal{P}'$; if we additionally have $v \in C_t'$, the construction of $\mathcal{P}'$ ensures that $v$ will also belong to this same block of $\mathcal{P}'$.

For the second condition, let $X$ be a block of $\mathcal{P}'$.  If $v \notin X$ then the condition is immediately satisfied due to the conditions for $\widehat{E}$ to belong to $\mathcal{E}^c(t,\sigma)$, so we may assume that $v \in X$.  If $X \setminus v \neq \emptyset$ then, by construction of $\mathcal{P}'$,  there is some vertex $u \in \mathcal{D}(t)$ such that $u$ and $v$ lie in the same component of $\widetilde{G_t}$ (note also that $u \in X$); hence the total number of vertices in connected components that intersect $X \setminus v$ is the same as the number in components that intersect $X$, and so as $X \setminus \{v\}$ is a block of $\mathcal{P}$ this number must be at most $c(X \setminus \{v\}) = c'(X)$. Finally, if $X = \{v\}$, the only connected component of $\widetilde{G_t}$ that intersects $X$ is the component $C$ that contains $v$, and by definition we have $c'(X) = |V(C)|$.  So the second condition is satisfied in all cases.

Thus we see that $\widehat{E} \in \mathcal{E}^c(t',\sigma')$, completing the argument that $\mathcal{E}^c(t,\sigma) \subseteq \bigcup_{\sigma' \in \sigma[t']_{\cfgt}} \left\lbrace \mathcal{E}^c(t',\sigma') \right\rbrace$ and hence the proof.
\end{proof}

\subsubsection*{Join nodes}

Suppose that the join node $t$ has children $t_1$ and $t_2$.  Note that $\cst(t_1) = \cst(t_2) = \cst(t)$.  In this case, we need to consider a set of pairs of inherited component-states for the two children $t_1$ and $t_2$.  We write $\cst(t_1,t_2)$ for the Cartesian product $\cst(t_1) \times \cst(t_2)$ and, given any component-state $\sigma = (\mathcal{P},c) \in \cst(t)$, we write $\sigma[t_1,t_2]_{\cjoin}$ for the set of pairs of \emph{join-inherited component-states} $((\mathcal{P}_1,c_1),(\mathcal{P}_2,c_2)) \in \cst(t_1,t_2)$ satisfying:
\begin{enumerate}
\item $\mathcal{P}_1 = \mathcal{P}_2 = \mathcal{P}$, and
\item for every block $X$ of $\mathcal{P}$, $c(X) = c_1(X) + c_2(X) - |X|.$
\end{enumerate}

We now make the following claim about the component-signature of $t$ and those of $t_1$ and $t_2$.
\begin{lma}
Let $t$ be a join node with children $t_1$ and $t_2$, and let $\sigma = (\mathcal{P},c) \in \cst(t)$.  Then
\begin{align*}
\cdel_k(t,\sigma) & = \left[ \min_{\substack{(\sigma_1,\sigma_2) \in \\ \sigma[t_1,t_2]_{\cjoin}}} \{\cdel_k(t_1,\sigma_1) + \cdel_k(t_2,\sigma_2) - |\{uw \in E(G[\mathcal{D}(t)]): u \not\sim_{\mathcal{P}} w\}| \right]_{\leq k}.
\end{align*}
\label{lma:join}
\end{lma}
\begin{proof}
As in the proof of Lemma \ref{lma:gen-join}, it suffices to prove that
$$\mathcal{E}^c(t,\sigma) = \bigcup_{(\sigma_1,\sigma_2) \in \sigma[t_1,t_2]_{\cjoin}} \lbrace E_1 \cup E_2: E_1 \in \mathcal{E}^c(t_1,\sigma_1), E_2 \in \mathcal{E}^c(t_2,\sigma_2) \rbrace.$$
To do this, we will exploit some simple observations about $V_t$, $V_{t_1}$ and $V_{t_2}$.  Note that $V_t = V_{t_1} \cup V_{t_2}$, and that by the properties of tree decompositions we have $V_{t_1} \cap V_{t_2} = \mathcal{D}(t)$.  Thus, for any set of vertices $U \subseteq V_t$, we have
\begin{equation}
|U| = |U \cap V_{t_1}| + |U \cap V_{t_2}| - |U \cap \mathcal{D}(t)|.
\label{join-vertex-set}
\end{equation}
Similarly, for any set of edges $F \subseteq E(G[V_t])$, we have
\begin{equation}
|F| = |F \cap E(G[V_{t_1}])| + |F \cap E(G[V_{t_2}])| - |F \cap E(G[\mathcal{D}(t)])|.
\label{join-edge-set}
\end{equation}

We first show that $\mathcal{E}^c(t,\sigma) \subseteq \bigcup_{(\sigma_1,\sigma_2) \in \sigma[t_1,t_2]_{\cjoin}} \lbrace E_1 \cup E_2: E_1 \in \mathcal{E}^c(t_1,\sigma_1), E_2 \in \mathcal{E}^c(t_2,\sigma_2) \rbrace$.  Let $\widehat{E} \in \mathcal{E}^c(t,\sigma)$, and let $\widetilde{G_t} = G[V_t] \setminus \widehat{E}$.  We set $\sigma_1 = (\mathcal{P},c_1)$ and $\sigma_2 = (\mathcal{P},c_2)$ where, for $X \in \mathcal{P}$ and $i \in \{1,2\}$, $c_i(X)$ is defined to be the total number of vertices in connected components of $\widetilde{G_t}[V_{t_i}]$ that intersect $X$.  It is then straightforward to verify that $(\sigma_1,\sigma_2) \in \sigma[t_1,t_2]_{\cjoin}$.  Moreover, if we set $E_1 = \widehat{E} \cap G[V_{t_1}]$ and $E_2 = \widehat{E} \cap G[V_{t_2}]$, then it follows easily that $E_1 \in \mathcal{E}^c(t_1,\sigma_1)$ and $E_2 \in \mathcal{E}^c(t_2,\sigma_2)$.  Since it is clear that $\widehat{E} = E_1 \cup E_2$, this shows that $\mathcal{E}^c(t,\sigma) \subseteq \bigcup_{(\sigma_1,\sigma_2) \in \sigma[t_1,t_2]_{\cjoin}} \lbrace E_1 \cup E_2: E_1 \in \mathcal{E}^c(t_1,\sigma_1), E_2 \in \mathcal{E}^c(t_2,\sigma_2) \rbrace$.  

Conversely, we will now demonstrate that $\bigcup_{(\sigma_1,\sigma_2) \in \sigma[t_1,t_2]_{\cjoin}} \lbrace E_1 \cup E_2: E_1 \in \mathcal{E}^c(t_1,\sigma_1), E_2 \in \mathcal{E}^c(t_2,\sigma_2) \rbrace \subseteq \mathcal{E}^c(t,\sigma)$.  To do this, fix $(\sigma_1,\sigma_2) \in \sigma[t_1,t_2]_{\cjoin}$, and let $E_1 \in \mathcal{E}^c(t_1,\sigma_1)$ and $E_2 \in \mathcal{E}^c(t_2,\sigma_2)$; setting $\widetilde{G_t} = G[V_t] \setminus (E_1 \cup E_2)$, we then need to demonstrate that $\widetilde{G_t}$ satisfies the three conditions of Definition \ref{def:comp-E}.  We also set $\widetilde{G_1} = G[V_{t_1}] \setminus E_1$ and $\widetilde{G_2} = G[V_{t_2}] \setminus E_2$.

Observe that, if $C$ is a connected component of $\widetilde{G_t}$ with $C \cap \mathcal{D}(t) = \emptyset$ then $C$ is contained entirely in either $\widetilde{G_1}$ or $\widetilde{G_2}$, and so $|C| \leq h$.  This fact, combined with the fact (demonstrated below) that condition 2 is satisfied, shows that condition 1(a) is satisfied. 

For condition 1(b), consider a connected component $C$ of $\widetilde{G_t}$, where $C_t = V(C) \cap \mathcal{D}(t) \neq \emptyset$.  We will argue that, for any two vertices $u,w \in V(C)$, we must have $u$ and $w$ belonging to the same block of $\mathcal{P}$.  Recall that $\sim_{\mathcal{P}}$ defines an equivalence relation on $\mathcal{D}(t)$.  Whenever there is a path from $u_1$ to $u_2$ in $\widetilde{G_1}$, $u_1$ and $u_2$ belong to the same component of $\widetilde{G_1}$ and hence $u_1 \sim_{\mathcal{P}} u_2$; similarly, if there is a path from $u_1$ to $u_2$ in $\widetilde{G_2}$, we have $u_1 \sim_{\mathcal{P}} u_2$.  Now consider $u,w \in V(C)$.  Since $u$ and $w$ belong to the same connected component of $\widetilde{G_t}$, there is a path $P$ from $u$ to $w$ in $\widetilde{G_t}$.  Let $u_1 = u, u_2, \ldots, u_r = w$ be the vertices of $V(P) \cap \mathcal{D}(t)$, listed in the order in which they occur as $P$ is traversed from $u$ to $w$. Recall from the properties of a tree decomposition that any path from a vertex in $V_{t_1} \setminus \mathcal{D}(t)$ to a vertex in $V_{t_2} \setminus \mathcal{D}(t)$ must pass through $\mathcal{D}(t)$.  Thus, for $1 \leq i \leq r-1$, it follows that the segment of $P$ from $u_i$ to $u_{i+1}$ is entirely contained in either $\widetilde{G_1}$ or $\widetilde{G_2}$, and hence that $u_i \sim_{\mathcal{P}} u_{i+1}$.  By transitivity, this implies that $u \sim_{\mathcal{P}} w$, or in other words that $u$ and $w$ belong to the same block of $\mathcal{P}$, as required.

For the second condition, let $X$ be a block of $\mathcal{P}$, and let $C_1,\ldots,C_r$ be the components of $\widetilde{G_t}$ that intersect $X$.  We want to show that 
$$\sum_{i=1}^r |V(C_i)| \leq c(X).$$
For each $i$, set $C_i^{(1)} = V(C_i) \cap V_{t_1}$ and $C_i^{(2)} = V(C_i) \cap V_{t_2}$.  Notice that $\{C_1 \cap X, \ldots, C_r \cap X\}$ is a partition of $X$, and moreover that $C_i \cap X = C_i^{(1)} \cap C_i^{(2)}$.  Hence
\begin{align*}
\sum_{i=1}^r |V(C_i)| & = \sum_{i=1}^r |C_i^{(1)} \cup C_i^{(2)}| \\
					  & = \sum_{i=1}^r \left(|C_i^{(1)}| + |C_i^{(2)}| - |C_i^{(1)} \cap C_i^{(2)}| \right) \\
					  & = \sum_{i=1}^r |C_i^{(1)}| + \sum_{i=1}^r |C_i^{(2)}| - \sum_{i=1}^r|C_i^{(1)} \cap C_i^{(2)}| \\
					  & = \sum_{i=1}^r |C_i^{(1)}| + \sum_{i=1}^r |C_i^{(2)}| - |X|,
\end{align*}
Moreover, it is clear that $\bigcup_{i=1}^r C_i^{(1)}$ is the set of vertices in components of $\widetilde{G_1}$ that intersect $X$, so $|\bigcup_{i=1}^r C_i^{(1)}| \leq c_1(X)$; since the sets $C_1^{(1)},\ldots,C_r^{(1)}$ are clearly disjoint, this implies that $\sum_{i=1}^r |C_i^{(1)}| \leq c_1(X)$.  Similarly, we have $\sum_{i=1}^r |C_i^{(2)}| \leq c_2(X)$.  Thus we see that
$$\sum_{i=1}^r |V(C_i)| \leq c_1(X) + c_2(X) - |X| \leq c(X),$$
by definition of $\sigma[t_1,t_2]_{\cjoin}$, so the second condition is satisfied.

This completes the argument that $\bigcup_{(\sigma_1,\sigma_2) \in \sigma[t_1,t_2]_{\cjoin}} \lbrace E_1 \cup E_2: E_1 \in \mathcal{E}^c(t_1,\sigma_1), E_2 \in \mathcal{E}^c(t_2,\sigma_2) \rbrace \subseteq \mathcal{E}^c(t,\sigma)$, and hence the proof.
\end{proof}

\subsection{Algorithmic formulation}
\label{sec:cpt-algorithmic}

In this section, we have expressed our lemmas for recursively calculating component-signatures at nodes in a tree decomposition as pseudocode (Algorithms \ref{alg:genStates} to \ref{alg:joinNode}).  For convenience, ``infinity" is a single arbitrarily large number; in an implementation this might be the maximum number possible in the programming system.

\begin{algorithm}
\caption{Algorithm generating the set of possible component-states $st(t)$ for bag $\mathcal{D}(t)$ }
\label{alg:genStates}
\begin{algorithmic}

\smallskip

\State {\bf Input}: A node $t$ of the nice tree decomposition $\mathcal{T}$, the bag at that node $\mathcal{D}(t)$, the graph $G$, integer $h$
\State {\bf Output}:  A set of component-states $st(t)$
\State ---------------------------------------------------------------------------------------------------
  \State $states \leftarrow \emptyset$ 
   \State $allPartition \leftarrow$ all partitions of $\mathcal{D}(t)$ such that each block is of size at most $h$
             \For {$\mathcal{P} \in allPartition$}
            \State $allFunctions \leftarrow $ all functions $c$ from $\mathcal{P} = \{X_1,\ldots X_{|\mathcal{P}|}\}$ to $[h]$ such that $c(X) \geq |X|$
            
            \For {$c \in allFunctions$ }
               \State add $(\mathcal{P}, c)$ to $states$
            \EndFor
       \EndFor       

   \State \textbf{return} $states$
    
  \end{algorithmic}
\end{algorithm}

\begin{algorithm}
\caption{Algorithm for finding component-signatures of leaf nodes}
\label{alg:leafNode}
\begin{algorithmic}

\smallskip

\State {\bf Input}: A leaf node $t$ of the nice tree decomposition $\mathcal{T}$, the bag at that node $\mathcal{D}(t)$, the graph $G$
\State {\bf Output}:  A mapping from component-states to $del$ values for component-states of the node
\State ---------------------------------------------------------------------------------------------------
   \State delValues $\leftarrow$ empty dictionary
   \State allStates $\leftarrow$ get all the valid $(\mathcal{P}, c)$ component-states using Algorithm \ref{alg:genStates}
   \For {$(\mathcal{P}, c)$ in allStates}
       \If {$|\{uw \in e(G[\mathcal{D}(t)]): u \not\sim_{\mathcal{P}} w\}| \leq k$}
          \State delValues[$t, (\mathcal{P}, c)$] $\leftarrow$ $|\{uw \in e(G[\mathcal{D}(t)]): u \not\sim_{\mathcal{P}} w\}|$
        \Else
             \State delValues[$t, (\mathcal{P}, c)$] $\leftarrow$ infinity
        \EndIf
   \EndFor
   \State \textbf{return} delValues

  \end{algorithmic}
\end{algorithm}

\begin{algorithm}
\caption{Algorithm for finding component-signatures of introduce nodes}
\label{alg:introduceNode}
\begin{algorithmic}

\smallskip

\State {\bf Input}: An introduce node $t$ of the nice tree decomposition $\mathcal{T}$, the bag at that node $\mathcal{D}(t)$, the graph $G$, the child $t'$ of $t$, and a table $delValuesChild$ of $\cdel_k$ values for component-states of $t'$
\State {\bf Output}:  A mapping from component-states to $\cdel_k$ values for component-states of the node
\State ---------------------------------------------------------------------------------------------------
   \State delValues $\leftarrow$ empty dictionary
   \State allStates $\leftarrow$ get all the valid $(\mathcal{P}, c)$ component-states using Algorithm \ref{alg:genStates}
   \For {$(\mathcal{P} = \{X_1,\ldots,X_r\}, c)$ in allStates}
%\\      
%      \State \textit{Comment: generating the $\sigma[t']_{intr}$}
%      \\
      \State  $\sigma[t']_{intr} \leftarrow \emptyset$
      \State $refinements \leftarrow$ all partitions of $X_r \backslash {v}$     
       \For{$\{Y_1,\ldots,Y_s\} \in refinements$ }  
           \State $\mathcal{P}' \leftarrow \{X_1,\ldots,X_{r-1},Y_1,\ldots,Y_s\}$
           \State $allCs \leftarrow$ set of all functions $c':\mathcal{P}'\rightarrow[h]$ such that $c'(X_i) = c(X_i)$ if $1 \leq i \leq r-1$ and $\sum_{i=1}^{s}{c'(Y_i)} = c(X_r) - 1$
            \For {$c' \in allCs$}
                \State add $(\mathcal{P}', c')$ to $\sigma[t']_{intr}$
            \EndFor
         \EndFor
      \State minValue $\leftarrow $ infinity
      \For {$\sigma' \in \sigma[t']_{intr}$}
          \State value $\leftarrow$ delValuesChild[$t', \sigma'$] + $|\{uv \in E(G[\mathcal{D}(t)]): u \not\sim_{\mathcal{P}} v\}|$
          \If {value $<$ minValue}
              \State minValue $\leftarrow$ value
          \EndIf
      \EndFor
      \\
       \If {$minValue \leq k$}
          \State delValues[$t, (\mathcal{P}, c)$] $\leftarrow$ $minValue$
        \Else
             \State delValues[$t, (\mathcal{P}, c)$] $\leftarrow$ infinity
        \EndIf
   \EndFor
   \State \textbf{return} delValues

  \end{algorithmic}
\end{algorithm}

\begin{algorithm}
\caption{Algorithm for finding component-signatures of forget nodes}
\label{alg:forgetNode}
\begin{algorithmic}

\smallskip

\State {\bf Input}: A forget node $t$ of the nice tree decomposition $\mathcal{T}$, the bag at that node $\mathcal{D}(t)$, the graph $G$, the child $t'$ of $t$, and a table $delValuesChild$ of $\cdel_k$ values for component-states of $t'$
\State {\bf Output}:  A mapping from component-states to $\cdel_k$ values for component-states of the node
\State ---------------------------------------------------------------------------------------------------
   \State delValues $\leftarrow$ empty dictionary
   \State allStates $\leftarrow$ get all the valid $(\mathcal{P}, c)$ component-states using Algorithm \ref{alg:genStates}
   \For {$(\mathcal{P}, c)$ in allStates}
       %  \State \textit{Comment: generating the $\sigma[t']_{fgt}$}
         \State $\sigma[t']_{fgt} \leftarrow \emptyset$
         \State $allPartitions \leftarrow $ all partitions $\mathcal{P}'$ such that $\mathcal{P} = \mathcal{P}' \backslash \{v\}$
        \For {$\mathcal{P}' \in allPartitions$}
         \State $allc' \leftarrow \emptyset$
         \State $c' \leftarrow $ empty function
         \State vSingleton $\leftarrow$ false
         \For {$Y \in \mathcal{P}'$}
            \If {$Y \backslash \{v\} \neq \emptyset$}
              \State $c'(Y) \leftarrow c(Y \backslash \{v\})$
            \Else 
            \State vSingleton $\leftarrow$ true
            \EndIf
         \EndFor
         \If {not vSingleton}
         	\State add $c'$ to $allc'$
         \Else
	         \For {$i=1$ to $h$}
         		\State $c'(\{v\}) \leftarrow i$
         		\State add $c'$ to $allc'$
         	\EndFor
         \EndIf
         \For {$\mathcal{P}' \in allPartitions$}
             	\For {$c' \in allc'$} 
                   \State add $(\mathcal{P}', c')$ to $\sigma[t']_{fgt}$
                \EndFor
          \EndFor 
        \EndFor
      \State minValue $\leftarrow $ infinity
      \For {$\sigma' \in \sigma[t']_{fgt}$}
          \State value $\leftarrow$ delValuesChild[$t', \sigma'$] 
          \If {value $<$ minValue}
              \State minValue $\leftarrow$ value
          \EndIf
      \EndFor
             \If {$minValue \leq k$}
          \State delValues[$t, (\mathcal{P}, c)$] $\leftarrow$ $minValue$
        \Else
             \State delValues[$t, (\mathcal{P}, c)$] $\leftarrow$ infinity
        \EndIf
   \EndFor
   \State \textbf{return} delValues

  \end{algorithmic}
\end{algorithm}

\begin{algorithm}
\caption{Algorithm for finding component-signatures of join nodes}
\label{alg:joinNode}
\begin{algorithmic}

\smallskip

\State {\bf Input}: A join node $t$ of the nice tree decomposition $\mathcal{T}$, the bag at that node $\mathcal{D}(t)$, the graph $G$, the join node's children $t_1, t_2$, and a table $delValuesChild$ of $\cdel_k$ values for component-states of $t_1$ and $t_2$
\State {\bf Output}:  A mapping from component-states to $\cdel_k$ values for component-states of the node
\State ---------------------------------------------------------------------------------------------------
   \State delValues $\leftarrow$ empty dictionary
   \State allStates $\leftarrow$ get all the valid $(\mathcal{P}, c)$ component-states using Algorithm \ref{alg:genStates}
   \For {$(\mathcal{P}, c)$ in allStates}
%   \\      
%      \State \textit{Comment: generating the $\sigma[t_1, t_2]_{join}$}
%      \\
      \State $ \sigma[t_1, t_2]_{join} \leftarrow \emptyset$ 
      \State $\mathcal{P}_1 \leftarrow \mathcal{P}$
      \State $\mathcal{P}_2 \leftarrow \mathcal{P}$
      \State $allFunctions \leftarrow $ set of all function pairs $(c_1, c_2)$ such that for every block $X \in \mathcal{P}$, $c(X) = c_1(X) + c_2(X) - |X|$
      \For {$(c_1, c_2) \in allFunctions$}
          \State add $((\mathcal{P}_1, c_1), (\mathcal{P}_2, c_2))$ to $\sigma[t_1, t_2]_{join}$
      \EndFor
       \State minValue $\leftarrow $ infinity
      \For {$(\sigma_1, \sigma_2) \in \sigma[t_1, t_2]_{join}$}
          \State value $\leftarrow$ delValuesChild[$t_1, \sigma_1$] + delValuesChild[$t_2, \sigma_2$] - $|\{uw \in E(G[\mathcal{D}(t)]): u \not\sim_{\mathcal{P}} w\}|$
          \If {value $<$ minValue}
              \State minValue $\leftarrow$ value
          \EndIf
      \EndFor
       \If {minValue $\leq k$}
          \State delValues[$t, (\mathcal{P}, c)$] $\leftarrow$ minValue
        \Else
             \State delValues[$t, (\mathcal{P}, c)$] $\leftarrow$ infinity
        \EndIf
   \EndFor
   \State \textbf{return} delValues

  \end{algorithmic}
\end{algorithm}

\subsection{Running time and extensions}
\label{sec:cpt-runtime}

At each of the $O(n)$ nodes of the nice tree decomposition, we will generate, and then iterate over, fewer than $(wh)^w$ component-states for that node.  For each of those component-states, we will need to consider a collection of at most $(wh)^w$ inherited component-states (or pairs of component-states, in the case of a join node) of the child node(s).  In the algorithm, we first generate each of the component-states for a given node, and the corresponding set of inherited component-states for its children, then iterate over each relevant combination of component-states, performing various constant-time operations.  Thus, at each of $O(n)$ nodes we do $O\left((wh)^{2w}\right)$ work, giving an overall time complexity of $O((wh)^{2w}n)$. 

For simplicity, we have only described the most basic version of the algorithm; however, it is straightforward to extend it to deal with more complicated situations, involving any or all of the following.
\\

\noindent
\emph{Deleting edges so that the sum of weights of vertices in any component is at most $h$, where a weight function $w:V(G) \rightarrow \mathbb{N}$ is given:} change condition 1(a) in the definition of $\mathcal{E}^c(t,\sigma)$ to $\sum_{v \in V(C)} w(v) \leq h$, and add to the definition of the set of valid component-states for a node the condition that, for each block $X$ of $\mathcal{P}$, we have $\sum_{v \in X} w(X) \leq c(X)$.\\

\noindent
\emph{Deleting edges so that each vertex $v$ belongs to a component containing at most $\ell(v)$ vertices, where a limit function $\ell: V(G) \rightarrow \mathbb{N}$ is given:} change condition 1(a) in the definition of $\mathcal{E}^c(t,\sigma)$ to $|V(C)| \leq \min_{v \in V(C)} \ell(v)$, and add to the definition of the set of valid component-states for a node the condition that, for each block $X$ of $\mathcal{P}$, we have $c(X) \leq \min_{v \in X} \ell(v)$.\\

\noindent
Neither of these adaptations changes the asymptotic running time of the algorithm.  It is also straightforward to deal with different deletion costs for different edges, and to output an optimal set of edges to delete, as discussed for the general algorithm in Section \ref{sec:gen-runtime}.

\subsection{Experimental results}
\label{sec:expt}

As an example of its practical use, we have tested an implementation of our algorithm and a constraint satisfaction programming (CP) formulation of the problem to find minimum deletions to maximum components of five vertices on graphs derived from the persistent Scottish trade links.

We report early preliminary results in Table \ref{tab:experiments}, giving the minimum deletions found by our algorithm and by the CP, and the time required.  For the CP, most of the deletions were not confirmed to be minimum within 2 hours, but we report the time to this confirmation if it occurred within the 2 hour running time limit.  
We use an ensemble CP approach, running two separate solvers, and reporting the best result within the two hour limit.  The first CP uses the MiniZinc \cite{minizinc} modelling language as a front-end to the Gecode solver \cite{schultemodeling,gecode}.   Gecode was chosen for its ease of use and fast performance: it participates regularly in benchmarking challenges, and came first in all categories in the MiniZinc Challenge in 2009, 2010, 2011, and 2012 \cite{stuckeyMinizinc}.  The second CP uses a custom-written program in the Choco solver \cite{choco}, a medalist in the last three MiniZinc Challenges, chosen for its speed and the presence of a local expert.  

Neither our algorithm nor the CP solver were successful in finding a minimum deletion within 24 hours on several of the graphs with higher treewidths (15 and 12), and we do not report these in the table.  
In all cases our algorithm is much faster in clock-time than the CP approach: both could likely be improved by careful optimisation of the implementations.

The implementation of our algorithm takes as input a tree decomposition (generated using  LibTW \cite{libtw}) of a graph and the graph itself. Our implementation of the tree decomposition-based method can be found at github.com/magicicada/fpt-edge-deletion.

\begin{table}
	\caption{Experimental results for deletion to a maximum component size of 5 comparing the performance of our tree decomposition based algorithm and the CP method.  $v(G)$ is the number of vertices in the graph, $e(G)$ the number of edges, and $tw(G)$ the treewidth.  For our tree decomposition-based method we report the minimum deletion found and the approximate clock time to find it.  For the CP method we report the minimum deletion found within 2 hours, the approximate clock time to find it, and the approximate clock time to confirm that deletion is minimum: a dash in this column indicates that the confirmation that the deletion was minimum did not complete within 2 hours. The graphs used are a selection of anonymised cattle trading graphs from Scotland. Times are given in seconds, with times over 30 minutes rounded up to the nearest hour.}
	\label{tab:experiments}
	\begin{tabular}{p{1.3cm}p{0.6cm}p{0.6cm}p{1.cm}|p{2cm}p{1cm}|p{2cm}p{1cm}p{1.7cm}}
%
%\hline
\multicolumn{4}{c}{Graph information} & \multicolumn{2}{c}{Tree decomposition method} & \multicolumn{3}{c}{CP method} \\
\hline
	Graph ID & v(G) & e(G) & tw(G) & Minimum deletion found & Time & Minimum deletion found & Time & Time to confirmation  \\
	\hline
2010-0 & 104 & 110 & 4 & 38 & 53  & 82 & 7200 &-\\ %7h37
2010-3 & 45 & 45 & 3 & 11 & 4  & 21 & 7200  &-\\ %5h55 & 
2010-4 & 38 & 38 & 3 & 20 & 7  & 24 & 3600 &-\\ % 7h 3m
2010-5 & 37 & 40 & 4 & 7 & 64 & 7 & 1032 & 1046 \\
2012-0 & 97 & 119 & 5 & 58 & 7200 & 93 & 7200  &-\\ %1 55
2012-1 & 72 & 74 & 3 & 20 & 11  & 37 & 7200  &-\\ %9h4
2012-4 & 31 & 30 & 2 & 12 & 3  & 12 & 1082  &-\\
2012-5 & 49 & 52 & 3 & 15 & 7 & 30  & 7200 &-\\% 5hr42
2013-1 & 45 & 47 & 4 & 20 & 35  & 25 & 3600 &-\\
2013-3 & 61 & 62 & 4 & 19 & 38  & 30 & 7200  &-\\ % 2h10m
2013-4 & 35 & 38 & 4 & 15 & 5  & 15 & 3600 &-\\ %1hr 10m
%2013-5 & 41 & 50 & 4 & 17 & 52  & - & - &-\\ 27 at 32 min 48
2013-6 & 39 & 41 & 3 & 14 & 4  & 17 & 7200 &-\\ %1h58
2014-0 & 32 & 49 & 4 & 28 & 445  & 28 & 3600 &-\\ %
2014-1 & 47 & 47 & 3 & 18 & 6  & 24 & 7200  &-\\ %1h 47m
2014-2 & 57 & 59 & 4 & 17 & 10  & 33 & 7200 &-\\ %5h 41m 
2014-3 & 41 & 40 & 2 & 15 & 5  & 18 & 3600 &-\\
2014-5 & 31 & 35 & 4 & 8 & 18  & 8 & 123 & 123\\
2014-6 & 48 & 49 & 3 & 21 & 11  & 29 &  7200 &-\\ %3h
2014-7 & 31 & 32 & 3 & 16 & 4  & 16 & 7 &-\\
\end{tabular}
\end{table}

\section{Conclusions and open problems}

We have investigated the relevance of the well-studied graph parameter treewidth to the structure of real-world animal trade networks, and have provided evidence that this parameter is likely to be small for many networks of interest for epidemiological applications.  Motivated by this observation, we have derived an algorithm to solve \genprobname ~on input graphs having $n$ vertices and treewidth bounded by some fixed constant $w$ in time \genruntime, if no graph in $\mathcal{F}$ has more than $h$ vertices.  The special case of this problem in which $\mathcal{F}$ is the set of all trees on at most $h+1$ vertices is of particular interest from the point of view of the control of disease in livestock, and we have derived an improved algorithm for this special case, running in time $O((wh)^{2w}n)$.  It is straightforward to adapt both algorithms to deal with more complicated situations likely to arise in the application.

Many open questions remain concerning the complexity of this problem more generally, as we are far from having a complete complexity classification.  We know that useful structure in the input graph is required to give an fpt-algorithm: we demonstrated that it is not sufficient to parameterise by the maximum component size $h$ alone (unless P=NP).  However, it remains open whether the problem might belong to FPT when parameterised only by the treewidth $w$; we conjecture that treewidth alone is \emph{not} enough, and that the problem is W[1]-hard with respect to this parameterisation.  Considering other potentially useful structural properties of input graphs, one question of particular relevance to epidemiology would be the complexity of the problem on planar graphs: this would be relevant for considering the spread of a disease based on the geographic location of animal holdings (in situations where a disease is likely to be transmitted between animals in adjacent fields).  

Furthermore, animal movement networks can capture more information on real-world activity when considered as \emph{directed} graphs, and the natural generalisation of \compprobname ~to directed graphs in this context would be to consider whether it is possible to delete at most $k$ edges from a given directed graph so that the maximum number of vertices \emph{reachable} from any given starting vertex is at most $h$.  Exploiting information on the direction of movements might allow more efficient algorithms for this problem when the underlying undirected graph does not have very low treewidth; a natural first question would be to consider whether there exists an efficient algorithm to solve this problem on directed acyclic graphs.

\section*{Acknowledgements}

The authors would like to thank the following: Ivaylo Valkov for his assistance in developing an initial implementation of this algorithm as part of a summer research project, and the Engineering and Physical Sciences Research Council for providing funding for this summer project; EPIC: Scotland's Centre of Expertise on Animal Disease Outbreaks, which supported JE for part of her work on this project; the Royal Society of Edinburgh which supported KM for part of her work on this project through a Personal Research Fellowship funded by the Scottish Government; Ciaran McCreesh and Patrick Prosser for their assistance in developing the CP formulation against which we compared the performance of our algorithm.

\end{document}